\documentclass{article}
\usepackage{algorithmic}
\usepackage{algorithm} 

\usepackage[margin = 2.5cm]{geometry}

\usepackage{amsmath, amssymb, amsthm}

\usepackage{graphicx}
\usepackage{color}
\usepackage{xcolor}

\usepackage{tikz}
\usetikzlibrary{arrows}

\usepackage{natbib}

\usepackage{nicefrac}
\usepackage{mathtools}
\newtheorem{theorem}{Theorem}[section]
\newtheorem{claim}[theorem]{Claim}
\newtheorem{corollary}[theorem]{Corollary}
\newtheorem{lemma}[theorem]{Lemma}
\newtheorem{definition}[theorem]{Definition}
\numberwithin{equation}{section}

\newcommand{\E}{\mathbf{E}}

\newcommand{\one}{\mathbf{1}}

\newcommand{\bbR}{\mathbb{R}}

\newcommand{\calE}{\mathcal{E}}

\newcommand{\calP}{\mathcal{P}}

\usepackage{wrapfig}
\usepackage{enumitem}
\newcommand{\ball}{\operatorname{Ball}}
\newcommand{\diam}{\operatorname{diam}}

\newcommand{\disagree}{\operatorname{dis}}
\newcommand{\ainf}{\operatorname{A}_{\infty}}
\newcommand{\given}{\,|\,}

\title{Local Correlation Clustering with Asymmetric Classification Errors\footnote{The conference version of this paper appeared in the proceedings of ICML 2021.}}
\author{
    Jafar Jafarov\thanks{Equal contribution. Jafar Jafarov and Yury Makarychev were supported by NSF CCF-1718820, CCF-1955173, and NSF TRIPODS CCF-1934843/CCF-1934813.
Sanchit Kalhan and Konstantin Makarychev were supported by
NSF CCF-1955351 and NSF TRIPODS CCF-1934931.} \\ University of Chicago 
 \and
    Sanchit Kalhan\footnotemark[2] \\ Northwestern University 
 \and
    Konstantin Makarychev\footnotemark[2] \\ Northwestern University 
 \and
    Yury Makarychev\footnotemark[2] \\ Toyota Technological Institute at Chicago
}
\date{}

\begin{document}
 \maketitle
 \begin{abstract}
In the Correlation Clustering problem, we are given a complete weighted graph $G$ with its edges labeled as ``similar" and ``dissimilar" by a noisy binary classifier. For a clustering $\mathcal{C}$ of graph $G$, a similar edge is in disagreement with $\mathcal{C}$, if its endpoints belong to distinct clusters; and a dissimilar edge is in disagreement with $\mathcal{C}$ if its endpoints belong to the same cluster. The disagreements vector, $\disagree$, is a vector indexed by the vertices of $G$ such that the $v$-th coordinate $\disagree_v$ equals the weight of all disagreeing edges incident on $v$. The goal is to produce a clustering that minimizes the $\ell_p$ norm of the disagreements vector for $p\geq 1$. We study the $\ell_p$ objective in Correlation Clustering under the following assumption: Every similar edge has weight in the range of $[\alpha\mathbf{w},\mathbf{w}]$ and every dissimilar edge has weight at least $\alpha\mathbf{w}$ (where $\alpha \leq 1$ and $\mathbf{w}>0$ is a scaling parameter). We give an $O\left((\nicefrac{1}{\alpha})^{\nicefrac{1}{2}-\nicefrac{1}{2p}}\cdot \log\nicefrac{1}{\alpha}\right)$ approximation algorithm for this problem. Furthermore, we show an almost matching convex programming integrality gap.
\end{abstract}

\section{Introduction}
Grouping objects based on the similarity between them is a ubiquitous and important task in machine learning.
%A ubiquitous task of fundamental importance in machine learning, is that of grouping objects based on the similarity between them. 
This similarity information between objects can be represented in many ways, some of them being pairwise distances between objects (objects which are closer are more similar) or the degree of similarity between pairs of objects (objects which are more similar have a higher degree of similarity).
\citet*{BBC04} introduced the Correlation Clustering problem, a versatile model that elegantly captures this task of grouping objects based on similarity information. Since its introduction, the correlation clustering problem has found use in a variety of applications, such as co-reference resolution (see e.g., \citet*{Cohen01learningto, Cohen02learningto}), spam detection (see e.g., \citet{Ramachandran07filteringspam}, \citet{BonchiGL14}), image segmentation (see e.g., \citet{Wirth10}) and multi-person tracking (see e.g., \citet{tang2016multi, tang2017multiple}). In the Correlation Clustering problem, we are given a set of objects with pairwise similarity information. Our goal is to partition the objects into clusters that agree with this information \textit{as much as possible}. The pairwise similarity information is given as a weighted graph $G$ with edges labeled as either ``positive/similar" or as ``negative/dissimilar" by a noisy binary classifier. For a clustering $\mathcal{C}$, a positive edge is in disagreement with $\mathcal{C}$, if its endpoints belong to distinct clusters; and a negative edge is in disagreement with $\mathcal{C}$ if its endpoints belong to the same cluster.

To ascertain the quality of the clustering produced, \citet{BBC04} studied the Correlation Clustering problem under two complimentary objectives. Over the years, the objective that has received the most attention is to find a clustering that minimizes the total weight of edges in disagreement. For the case of complete unweighted graphs, \citet{BBC04} gave a constant factor approximation algorithm for this objective. \citet*{ACN08} improved the approximation ratio to $3$ by presenting a simple-yet-elegant combinatorial algorithm. They also presented a $2.5$-approximation algorithm based on Linear Programming (LP) rounding which was later derandomized without any loss in approximation ratio by~\citet*{vZHJW}. Finally, \citet*{CMSY15} gave an LP rounding algorithm which improved the approximation ratio to $2.06$. The standard LP was shown to have an integrality gap of $2$ by \citet*{CGW03} for the case of complete unweighted graphs. For the case of general weighted graphs, \citet{CGW03} and \citet*{DEFI06} gave an $O(\log n)$-approximation algorithm.

% This paragraph is a duplicate.
%For the case of complete unweighted graphs, the state-of-the-art is a $2.06 - \eps$-approximation algorithm based on Linear Programming (LP) rounding by \citet*{CMSY15}. For the case of general weighted graphs, \citet{CGW03} and \citet*{DEFI06} gave a $O(\log n)$-approximation algorithm.
%\knote{Briefly mention known results for weighted unweighted graphs somewhere; maybe here.}

Define the disagreements vector to be a vector indexed by the vertices of $G$. Given a clustering $\calP$, $\disagree(\calP, E^+, E^-) \in \bbR^V$ is a $|V|$-dimensional vector
where the $u$-th coordinate is equal to the weight of disagreements at $u$ with respect to $\calP$. That is, $\disagree_u(\calP, E^+, E^-)= \sum\limits_{(u,v) \in E}w_{uv}\cdot \one\{(u,v)\text{ is in disagreement with }\calP\}.$
%\begin{align*}
%&\disagree_u(\calP, E^+, E^-) \\
%&= \sum\limits_{(u,v) \in E} \one\{(u,v)\text{ is in disagreement with }\calP\}
%\end{align*}
Thus, minimizing the total weight of disagreements is equivalent to finding a clustering minimizing the $\ell_1$ norm of the disagreements vector. Another objective for Correlation Clustering that has received attention recently is to minimize the weight of disagreements at the vertex that is worst off (also known as Min Max Correlation Clustering). This is equivalent to finding a clustering that minimizes the $\ell_{\infty}$ norm of the disagreements vector. Observe that minimizing the $\ell_1$ norm is a global objective since the focus is on minimizing the total weight of disagreements. In contrast, for higher values of $p$ (particularly $p=\infty$), minimizing the $\ell_{p}$ norm becomes a more local objective since the focus shifts towards minimizing the weight of disagreements at a single vertex. Minimizing the $\ell_2$ norm of the disagreements vector can thus provide a balance between these global and local perspectives -- it considers the weight of disagreements at all vertices but penalizes vertices that are worse off more heavily. The following scenario is a showcase that minimizing the $\ell_2$ norm might be a more suitable objective than minimizing the $\ell_1$ norm. Consider a recommender system such that input is a bipartite graph with left and right sides representing customers and services, respectively. A positive edge implies that a customer is satisfied with the service; whereas a negative edge implies that they are dissatisfied with or have not used the service. We may be interested in grouping customers and services so that the total and the individual dissatisfaction of customers are minimized.

\begin{definition}\label{def:local}\textbf{(Local Correlation Clustering)}
Given an instance of Correlation Clustering $G = (V, E = E^+ \cup E^-)$ and $p\geq 1$, the local objective is to find a partitioning $\calP$ that minimizes the $\ell_p$ norm.
%
%\footnote{We use the standard definition of the $\ell_p$ norm of a vector $x$: $\|x\|_p = \left(\sum_u |x_u|^p\right)^\frac{1}{p}$.}  of the disagreements vector $\|\disagree(\calP, E^+, E^-)\|_p$.%
\end{definition}

We use the standard definition of the $\ell_p$ norm of a vector $x$: $\|x\|_p = \left(\sum_u |x_u|^p\right)^\frac{1}{p}$. Since its introduction by~\citet*{PueMil18}, local objectives for Correlation Clustering have been mainly studied under two models (see~\citet*{CharikarGS17}, \citet*{AhmadiKS19}, \citet*{KalhanMZ19}). We will refer to these models as (1) Correlation Clustering on Complete Graphs, and (2) Correlation Clustering with Noisy Partial Information. In the first model, the input graph $G$ is complete and unweighted. For this model, the first approximation algorithm was by~\citet{PueMil18} with an approximation factor of $48$ for minimizing the $\ell_p$ norm. This was later improved to $7$ by~\citet{CharikarGS17}. Lastly,~\citet{KalhanMZ19} provided a $5$ approximation algorithm. In the second model, $G$ is an arbitrary weighted graph with possibly missing edges. For minimizing the $\ell_{\infty}$ norm of the disagreements vector in this model, \citet{CharikarGS17} provided a $O(\sqrt{n})$ approximation. \citet{KalhanMZ19} gave an $O(n^{\frac{1}{2}-\frac{1}{2p}}\cdot \log^{\frac{1}{2}+\frac{1}{2p}}n)$-approximation algorithm for minimizing the $\ell_p$ norm of the disagreements vector.

We study local objectives in a different model -- Correlation Clustering with Asymmetric Classification Errors -- recently introduced by~\citet*{JafarovKMM20}. In this model, the input graph $G$ is complete and weighted. Furthermore, the ratio of the smallest edge weight to the largest positive edge weight is at least $\alpha\leq 1$. Thus, for some $\mathbf{w}>0$, each positive edge weight lies in the interval $[\alpha\mathbf{w},\mathbf{w}]$ and each negative edge weight is at least $\alpha\mathbf{w}$. This model better captures the subtleties in real world instances than the standard models. Since real world instances rarely have equal edge weights, assumptions in the Correlation Clustering on Complete Graphs model are too strong. In contrast, in the Correlation Clustering with Noisy Partial Information model, we can have edge weights that are arbitrarily small or large, an assumption which is too weak. In many real world instances, the edge weights lie in some range $[a,b]$ with $a,b >0$. For this model, \citet{JafarovKMM20} gave a $(3+2\ln\frac{1}{\alpha})$ approximation for minimizing the $\ell_1$ norm of the disagreements vector.
\begin{definition}\label{def:alphaC}
Correlation Clustering with Asymmetric Classification Errors is a variant of Correlation Clustering on Complete Graphs. We assume that the weight of each positive edge lies in $[\alpha\mathbf{w},\mathbf{w}]$ and the weight of each negative edge lies in $[\alpha\mathbf{w},\infty)$, where $\alpha \in (0,1]$ and $\mathbf{w} >0$.
\end{definition}

\paragraph{Our Contributions.}
In this paper we study the task of minimizing local objectives (Definition~\ref{def:local}) under the Correlation Clustering with Asymmetric Classification Errors model (Definition~\ref{def:alphaC}).
Our main result is an $O\left(\left(\frac{1}{\alpha}\right)^{\frac{1}{2}-\frac{1}{2p}}\cdot \log\frac{1}{\alpha}\right)$ approximation algorithm for minimizing the $\ell_p$ norm of the disagreements vector, which we now state.
\begin{theorem}\label{thm:main}
There exists a polynomial-time $O\left(\left(\frac{1}{\alpha}\right)^{\frac{1}{2}-\frac{1}{2p}}\cdot \log\frac{1}{\alpha}\right)$-approximation algorithm for the $\ell_p$ objective in the Correlation Clustering with Asymmetric Classification Errors model.
\end{theorem}
For $p=1$, our algorithm provides an $O(\log\frac{1}{\alpha})$ approximation, which matches the approximation guarantee given by \citet{JafarovKMM20} up to constant factors.
Consider $p=2$, that is, the $\ell_2$ norm. If we ignored the edge weights and applied the state of the art algorithm in the Correlation Clustering on Complete Graphs model, we would get an $O(\frac{1}{\alpha})$ approximation. If we were to use the state of the art algorithm in the Correlation Clustering with Noisy Partial Information model, we would get an $\Tilde{O}\left(n^{\nicefrac{1}{4}}\right)$ approximation. However, by using our algorithm (Theorem~\ref{thm:main}), we obtain an $\Tilde{O}\big(\big(\nicefrac{1}{\alpha}\big)^{\nicefrac{1}{4}}\big)$ approximation, which is a huge improvement when $\nicefrac{1}{\alpha}\ll n$. 

\begin{corollary}
There exists a polynomial-time $O\big(\big(\nicefrac{1}{\alpha}\big)^{\nicefrac{1}{4}}\cdot \log\frac{1}{\alpha}\big)$-approximation algorithm for the $\ell_2$ objective in the Correlation Clustering with Asymmetric Classification Errors model.
\end{corollary}

Finally, we present the implication of our main result for the $\ell_\infty$ norm. For the $\ell_\infty$ norm, \citet{KalhanMZ19} presented an $\Tilde{O}(\sqrt{n})$ approximation under the Correlation Clustering with Noisy Partial Information model. Using our algorithm for Correlation Clustering with Asymmetric Classification Errors we obtain an $\Tilde{O}(\sqrt{\nicefrac{1}{\alpha}})$-approximation factor, which is a significant improvement to the approximation guarantee in this setting. 

\begin{corollary}
There exists a polynomial-time $O(\sqrt{\nicefrac{1}{\alpha}} \cdot \log \nicefrac{1}{\alpha})$-approximation algorithm for the $\ell_\infty$ objective in the Correlation Clustering with Asymmetric Classification Errors model.
\end{corollary}

We emphasize that our approximation ratio for the $\ell_p$ norm is independent of the graph size and only depends on $\alpha$. % -- the ratio of the lightest edge to the heaviest positive edge.

Our algorithm relies on the natural convex programming relaxation for this problem (Section~\ref{sec:relaxation}). We compliment our positive result (Theorem~\ref{thm:main}) by showing that it is likely to be the best possible based on the natural convex program, by providing an almost matching integrality gap. %We prove this theorem in Appendix~\ref{Gap_section}.
\begin{theorem}\label{thm:int_gap}
The natural convex programming relaxation for the $\ell_p$ objective in the Correlation Clustering with Asymmetric Classification Errors model has an integrality gap of $\Omega\big(\big(\nicefrac{1}{\alpha}\big)^{\frac{1}{2}-\frac{1}{2p}}\big)$.
\end{theorem}

\noindent \textbf{Organization of the paper.} In Section~\ref{sec:relaxation}, we describe the convex relaxation that we will use in our algorithm for Correlation Clustering. In Section~\ref{sec:overview}, we introduce a novel technique for partitioning metric spaces. This forms the main technical basis for our algorithm for Correlation Clustering. In Section~\ref{sec:algorithm}, we prove our main result, Theorem~\ref{thm:main}. In Section~\ref{sec:proof_overview}, we describe our metric space partitioning scheme and give a proof overview of its correctness. In sections~\ref{sec:itr-clustering},~\ref{sec:Cluster-select},~\ref{sec:light-ball} and~\ref{sec:proof-lem-Phi-S}, we formally prove the correctness of our partitioning scheme, Theorem~\ref{thm: decompositon}. In Appendix~\ref{Gap_section}, we prove our integrality gap result, Theorem~\ref{thm:int_gap}.
 
\section{Convex Relaxation}\label{sec:relaxation}

Our algorithm for minimizing local objectives is based on rounding the optimal solution to a suitable convex program (Figure~\ref{fig:CP}). This convex program is similar to the relaxations used in \citet{CharikarGS17} and \citet{KalhanMZ19}.
In this convex program,  we have a variable $x_{uv}$ for every pair of vertices $u,v\in V$. The variable $x_{uv}$ captures the distance between $u$ and $v$ in the ``multicut metric''. 
In the integral solution, $x_{uv} = 0$ if $u$ and $v$ are in the same partition and $x_{uv} = 1$ if $u$ and $v$ are in different partitions. In order to enforce that the partitioning is consistent, we add triangle inequality constraints between all triplets of vertices (P2). We also require that distance $x_{uv}$ is symmetric (P3).

%%That is, $x_{uv} = 0$ denotes that $u$ and $v$ are in the same partition, while $x_{uv} = 1$ denotes that $u$ and $v$ are in different partitions. In order to enforce that the %%partitioning is consistent, we ensure that the triangle inequality is satisfied by every triplet of vertices (P2). Moreover, we ensure that the variables are symmetric (P3).

For every vertex $u \in V$, we use the variable $y_u$ to denote the total weight of violated edges incident on $u$ (P1). The objective of the convex program is thus to minimize $\Vert y \Vert_p$ -- the $\ell_p$ norm of the vector $y$. Notice that each constraint in the convex program is linear, and the objective function $\|\cdot\|_p:\bbR^n \to \bbR$ is convex (by the Minkowski inequality). 

\begin{figure*}
  \centering
\begin{equation*}
\begin{array}{ll@{}llr}
\text{minimize}  & \displaystyle \|y\|_p & &\text{(P)}\\
\text{subject to}
& y_u=\displaystyle\sum_{v:(u,v)\in E^+} w_{uv} x_{uv} + \sum_{v:(u,v)\in E^-} w_{uv} (1 - x_{uv}) \quad &\text{for all } u \in V & \text{(P1)}\\
&x_{v_1v_2}+x_{v_2v_3} \geq x_{v_1v_3} \quad &\text{for all } v_1,v_2,v_3 \in V&\text{(P2)}\\
&x_{uv} = x_{vu}\quad &\text{for all } u,v \in V& \text{(P3)}\\
& x_{uv} \in [0,1] &\text{for all } u,v \in V & \text{(P4)}\\
\end{array}
\end{equation*}
  \caption{Convex relaxation for Correlation Clustering with min $\ell_p$ objective for $p\geq 1$ or $p = \infty$.}\label{fig:CP}
\end{figure*}

What remains to be shown is that the relaxation presented in Figure~\ref{fig:CP} is valid. To this end, consider any partition $\mathcal{P} = (P_1, P_2, \ldots, P_k)$ of the set of vertices $V$. For every pair of vertices $u,v$, if $u$ and $v$ lie in the same partition, we assign the corresponding variable $x_{uv}$ a value of $0$, else we assign it a value of $1$. Note that such an assignment satisfies the triangle inequality (P2). Variable $y_u$ thus captures the total weight of violated edges incident on $u$; every similar edge $(u,v)$ incident on $u$ that crosses a partition contributes $w_{uv} \cdot x_{uv} = w_{uv}$ to $y_u$, and every dissimilar edge present within a cluster contributes $w_{uv} \cdot (1 - x_{uv}) = w_{uv}$ to $y_u$. Thus, $y_u$ is equal to $\disagree_u(\mathcal{P}, E^+,E^-)$. %\jnote{disagree vector is not defined formally} 
Hence, an integral convex program solution defined in such a manner is feasible and has the same cost as the partitioning. It is possible, however, that the cost of the optimal fractional solution is less than the cost of the optimal integral solution, and hence the convex program in Figure~\ref{fig:CP} is a relaxation to our problem. We note that our relaxation is simpler than the relaxation used in~\citet{KalhanMZ19}. The additional variables in their convex program are not needed in our case because all edge weights belong to the interval $[\alpha\mathbf{w},\mathbf{w}]$.
\section{A New Technique for Partitioning Metric Spaces}\label{sec:overview}
%\section{Overview of Our Approach}
%\input{figures/algo_CC}
%The proof of our main theorem, Theorem~\ref{thm:main}, is broken down into several steps. In this section, we describe an overview of our approach. 

We will use the following notation: Given expressions $X$ and $Y$, we write $X \lesssim Y$ if $X \leq C \cdot Y$ for some constant $ C > 0$ (that is, $X = O(Y)$). 
We define $\gtrsim$ similarly. Furthermore, let $X^+=0$ if $X < 0$ and $X^+=X$ if $X \geq 0$.
We use $\ball(v, l) = \{u : d(u,v) \leq l\}$ to denote the set of vertices at a distance of at most $l$ from $v$. 

In this section, we describe our main technical tool -- a novel probabilistic scheme for partitioning metric spaces which may be of independent interest. This partitioning scheme forms the basis of our algorithm (Algorithm~\ref{alg:CC}) for Correlation Clustering. We begin by stating this technical result.

% \subsection {Algorithm for Correlation Clustering}
% Our algorithm (Algorithm~\ref{alg:CC}) for minimizing local objectives for Correlation Clustering with Asymmetric Classification Errors begins by solving the convex relaxation 
% in Figure~\ref{fig:CP} to obtain solution $\{x_{uv}\}_{u,v \in V}$. It then defines a metric $d(\cdot,\cdot)$ on $V$ by setting distances $d(u, v) = x_{uv}$ and
% defines parameters $r = \frac{\sqrt{\alpha}}{\ln (\nicefrac{1}{\alpha})}$ and $ R = \nicefrac{1}{3}$.
% %\footnote{Thus, $\beta = \nicefrac{r}{R} =  \frac{3\sqrt{\alpha}}{\log (\nicefrac{1}{\alpha})}$ and $D_\beta =  O(\log\nicefrac{1}{\beta})) = O(\log\nicefrac{1}{\alpha})$)}.
% At this point, the algorithm makes use of our key technical contribution -- a new probabilistic scheme for partitioning metric spaces -- and outputs the partitioning thus obtained.
% Please refer to Algorithm~\ref{alg:CC} for a summary. We now state our main technical result.

\begin{theorem}\label{thm: decompositon}
For every $q \geq 1$ there exists a $\beta^*_q=\Theta\big(\frac{1}{q\ln(q+1)}\big) < 1$ such that the following holds. Consider a finite metric space $(X,d)$. 
Fix two positive numbers $r$ and $R$ such that $\beta = r/R\leq \beta^*_q$. 
Let $D_{\beta}=2(q+1)\ln\nicefrac{1}{\beta}$. 
Then, there exists a probabilistic partitioning $\calP$ satisfying properties (1), (2), and (3):

\begin{enumerate}
    \item[(1)] $\diam(P)\leq 2R\mbox{ }$ for every $P\in \calP$ (always); 
    
    \item[(2)]  For every point $u$ in $X$, the following bound holds:
    $$\sum_{v\in \ball(u,R)}\bigg(\Pr\big\{\calP(u)\neq \calP(v)\big\} - 
    D_{\beta}\frac{d(u,v)}{R}\bigg)^+
    \lesssim
    \beta^q
    \sum_{v\in \ball(u,2R)} \frac{d(u,v)}{R},$$
    where $\calP(u)$ denotes the partition of $\calP$ that contains $u$.
    
    \item[(3)] Moreover, for every $u$ in $X$, we always have, 
    $$
    \sum_{v\in \ball(u,r)}\one\big\{\calP(u)\neq \calP(v)\big\}\lesssim
    \beta \cdot D^2_{\beta}
    \sum_{v\in\ball(u,2R)}\frac{d(u,v)}{R}.
    $$
\end{enumerate}
\end{theorem}

The partitioning we construct in Theorem~\ref{thm: decompositon} resembles a $2D$–separating $2R$-bounded stochastic
decomposition of a metric space~\cite{bartal1996,CKR,FRT}. Recall that a $2D$–separating $2R$-bounded stochastic
decomposition satisfies property (1) of Theorem~\ref{thm: decompositon} and the $2D$-separating condition: for every $u,v\in X$,
\begin{equation}\label{eq:separating-condition}
\Pr\big\{\calP(u)\neq \calP(v)\big\} - D\frac{d(u,v)}{R} \leq 0.
\end{equation}

At a very high level, the goals of our partitioning and the $2D$–separating $2R$-bounded stochastic
decomposition are similar: decompose a metric space in clusters of diameter at most $2R$ so that nearby points lie in the same cluster with high enough probability.  However, the specific conditions are quite different. Loosely speaking, property (2) of Theorem~\ref{thm: decompositon} says that the decomposition satisfies (\ref{eq:separating-condition}) with $D = D_\beta$  on average up to an additive error term of $O(\beta^q)\sum_{v\in \ball(u,2R)} \frac{d(u,v)}{R}$. Crucially, property (3) provides an analogous guarantee not only in expectation, but also in the worst case (which a $2D$-separating decomposition does not satisfy). 

Property (3) plays a key role in proving our main result, Theorem~\ref{thm:main}. For the standard objective function for Correlation Clustering (minimizing the $\ell_1$ norm of the disagreements vector), properties (1) and (2) are sufficient since an upper bound on the expected weight of disagreements on a single vertex implies an upper bound on the expected weight of the total disagreements. The situation gets trickier when we consider minimizing arbitrary $\ell_p$ ($p >1$) norms of the disagreements vector. For instance, having an upper bound on the expected weight of disagreements on a single vertex does not necessarily translate to an upper bound on the expected weight of disagreements on a worst vertex ($\ell_\infty$ norm). We overcome this nonlinear nature of the problem for higher values of $p$ by using the deterministic (worst-case) guarantee given by property (3) of Theorem~\ref{thm: decompositon}.

Also note that coefficients $D_\beta$ and $\beta$ do not depend on the size $|X|$ of the metric space (in our algorithm, they will only depend on $\alpha$, which is defined as the ratio of the smallest edge weight to the largest positive edge weight). However, the optimal value of $D$ in the $2D$-separating condition is $\Theta(\log |X|)$. %\snote{We make a reference (the first time in this section) to $\alpha$ here. Should we remind the reader about its definition?}\jnote{Added the reminder.}

\section{Correlation Clustering via Metric Partitioning}\label{sec:algorithm}
In this section, we will prove our main theorem, Theorem~\ref{thm:main}. Our algorithm (Algorithm~\ref{alg:CC}) for minimizing local objectives for Correlation Clustering with Asymmetric Classification Errors begins by solving the convex relaxation in Figure~\ref{fig:CP} to obtain a solution $\{x_{uv}\}_{u,v \in V}$. It then defines a metric $d(\cdot,\cdot)$ on $V$ by setting distances $d(u, v) = x_{uv}$.
%and parameters $r = \frac{\sqrt{\alpha}}{\ln (\nicefrac{1}{\alpha})}$, $ R = \nicefrac{1}{3}$ and $q=2$.

We let $q = 2$. Let $\alpha^*$ be the solution of equation $3\sqrt{\alpha^*}/\ln\nicefrac{1}{\alpha^*} = \beta^*_2$ (note that $\alpha^*$ is an absolute constant). 
We assume that $\alpha \leq \alpha^*$. If $\alpha > \alpha^*$, 
we just redefine $\alpha$ as $\alpha^*$ (this will increase the approximation ratio only by a constant factor).
 We set $r={\sqrt{\alpha}}/{\ln \nicefrac{1}{\alpha}}$ and $R=\nicefrac{1}{3}$. 
 Note that $r / R \leq \beta^*_2 < 1$.

 At this point, the algorithm makes use of our key technical contribution -- a new probabilistic scheme for partitioning metric spaces (Algorithm~\ref{alg:MPSP}) -- and outputs the partitioning thus obtained.
 Please refer to Algorithm~\ref{alg:CC} for a summary.
 
 To show that $\calP$ has the desired approximation ratio in Theorem~\ref{thm:main}, we bound the weight of disagreements at every vertex $u \in V$ with respect to $\calP$.
 To this end, we show that two useful quantities, the total weight of disagreements at $u$ and the expected weight of disagreements at $u$ can be bounded in terms of $y_u$, the cost paid by the convex program for vertex $u$. In Theorem~\ref{thm:local_guarantee}, we make use of the properties of $\calP$ given by Theorem~\ref{thm: decompositon}
 to get a bound on these two quantities for each vertex $u \in V$. Then, in Section~\ref{subsec:prfMain}, we use the bounds from Theorem~\ref{thm:local_guarantee} to complete the proof of  Theorem~\ref{thm:main}: we show that if the total cost of disagreements and the expected cost of disagreements with respect to $\calP$ are bounded for every $u \in V$, then the 
 partitioning $\calP$ achieves the desired approximation ratio in Theorem~\ref{thm:main}. We remind the reader that given a partitioning $\calP$ of the vertex set and a vertex $u \in V$,  $\disagree_u(\calP, E^+, E^-)$ denotes the weight of edges incident on $u$ that are in disagreement with respect to $\calP$. 
Moreover, $y_u$ denotes the convex programming (CP) cost of the vertex $u$. 

Define $A_1 = \ln \nicefrac{1}{\alpha}$ and $\ainf = \nicefrac{\ln (\frac{1}{\alpha})}{\sqrt{\alpha}} = \nicefrac{1}{r}$. Our analysis focuses on bounding two key quantities related to a vertex $u \in V$. The first quantity, $\disagree_u(\calP, E^+, E^-)$, is the total weight of edges incident on $u$ that are in disagreement with $\calP$. We show that this quantity can be charged to the CP cost of $u$ and is at most $\ainf \cdot y_u$. We then get a stronger bound for our second quantity of interest, $\E[\disagree_u(\calP, E^+, E^-)]$, the expected cost of a vertex $u$. In particular, we show that $\E[\disagree_u(\calP, E^+, E^-)] \leq A_1 \cdot y_u$.

\begin{algorithm}[tb]
   \caption{Correlation Clustering Algorithm}
   %\caption{Minimizing Local Objectives for Correlation Clustering with Asymmetric Classification Errors.}
   \label{alg:CC}
\begin{algorithmic}
   \STATE {\bfseries Input:} $G = (V, E^+, E^-,\mathbf{w}, \alpha)$, $\{x_{uv}\}_{u,v \in V}$.
   %\STATE {\bfseries Input:} Graph $G = (V, E= (E^+\cup E^-),\mathbf{w}, \alpha)$ where,
   %\STATE $w_e \in [\alpha \mathbf{w}, \mathbf{w}]$ for $e \in E^+$
   %\STATE $w_e \in [\alpha\mathbf{w}, \infty)$ for $e \in E^-$
   %\STATE Let $\{x_{uv}\}_{u,v \in V}$ be the solution to Convex Program~\ref{fig:CP} for instance $(V,E, \mathbf{w}, \alpha)$
   \STATE Define a metric $d$ on $V$ such that $d(u,v) = x_{uv}$ for all $u,v\in V$.
   %\STATE Define $Y=V$, and distances $d(u,v) = d(v,u) = x_{uv}$.
   \STATE Define $r = (\nicefrac{\sqrt{\alpha}}{\ln \nicefrac{1}{\alpha}}), R = \nicefrac{1}{3}$, $q=2$.
   \STATE $\calP =$ Metric Space Partitioning Scheme$(V,d,r,R, q)$.
   %\STATE $\calP =$ Metric Space Partitioning Scheme$(X,d,r,R)$.
   \STATE Output $\calP$.
\end{algorithmic}
\end{algorithm}                    
\begin{theorem}\label{thm:local_guarantee}
Given an instance of Correlation Clustering with Asymmetric Classification Errors (Definition~\ref{def:alphaC}), Algorithm~\ref{alg:CC} outputs a partitioning $\calP$ of the vertex set such that the following holds for every vertex $u \in V$:
\begin{itemize}
	\item[(a)] $\disagree_u(\calP, E^+, E^-) \lesssim A_\infty \cdot y_u$;
	\item[(b)] $\E[\disagree_u(\calP, E^+, E^-)] \lesssim A_1 \cdot y_u$,
\end{itemize}
where $A_1 = \ln (\nicefrac{1}{\alpha})$ and $\ainf = \nicefrac{\ln (\frac{1}{\alpha})}{\sqrt{\alpha}}$.
\end{theorem}
\begin{proof}
Without loss of generality we assume that the scaling parameter $\mathbf{w}$ is $1$. Thus, for every positive edge $e^+ \in E^+$, $w_{e^+} \in [\alpha, 1]$, while for every negative edge $e^- \in E^-, w_{e^-} \geq \alpha$.
Write the formula for $\disagree_u(\calP, E^+, E^-)$ for a given vertex $u \in V$, 

\begin{align*}
    \disagree_u&(\calP, E^+, E^-) = \displaystyle\sum_{(u,v)\in E^+} w_{uv} \cdot \one\{\calP(u) \neq \calP(v)\} 
    + \displaystyle\sum_{(u,v)\in E^-} w_{uv} \cdot \one\{\calP(u) = \calP(v)\}.
\end{align*}

Let $E^{\geq r}$ be the set of positive edges $(v,w)$ in $E^+$ with $x_{vw} \geq r$. Observe that

\begin{align}\label{eqn:disagreement_parsed}
 \disagree_u(\calP, E^+, E^-) = &\phantom{+}\disagree_u(\calP, \varnothing, E^-)
  +\disagree_u(\calP, E^{\geq r}, \varnothing)
  + \disagree_u(\calP, E^+\setminus E^{\geq r}, \varnothing).
\end{align}
Recall that 
$\beta=r/R = \nicefrac{3\sqrt{\alpha}}{\ln \nicefrac{1}{\alpha}}$, $q=2$, and $D_\beta = \Theta\left(\ln\nicefrac{1}{\beta}\right) = \Theta\left(\ln\nicefrac{1}{\alpha}\right)$.
From Theorem~\ref{thm: decompositon}, part (a), 
we know that the diameter of each partition $P$ in $\calP$ is at most $2R$. For any negative edge to be in disagreement, 
both its endpoints must lie in the same partition. Thus, the length $x_{uv}$ for any such edge $(u,v) \in E^-$ is at most $2R$, and hence its CP contribution is at most $(1 - 2R) = \nicefrac{1}{3}$.
Hence, 
$$\disagree_u(\calP, \varnothing, E^-) = \sum\limits_{(u,v)\in E^-} w_{uv} \one\{\calP(u) = \calP(v)\} \leq 3 y_u.$$ 
Then,
$$\disagree_u(\calP, E^{\geq r}, \varnothing) \leq |\{v: (u,v)\in E^{\geq r}\}| \leq \frac{y_u}{r} = \ainf y_u.$$ 

To complete the proof of Theorem~\ref{thm:local_guarantee}, part (a) we write:
\begin{align*}
    \disagree_u(\calP, E^+\setminus E^{\geq r}, \varnothing) 
    &= \sum_{v \in \ball(u, r)} w_{uv} \cdot \one\big\{\calP(u) \neq \calP(v)\big\} \\
    &\leq \sum_{v \in \ball(u, r)} \one\big\{\calP(u) \neq \calP(v)\big\}.
\end{align*}
The inequality above holds because the weight of each positive edge is at most $1$. Next, using the bound for $\sum_{v \in \ball(u, r)} \one\big\{\calP(u) \neq \calP(v)\big\}$ from Theorem~\ref{thm: decompositon} part (c), we get,
\begin{align*}
    \sum_{v \in \ball(u, r)} \one\big\{\calP(u) \neq \calP(v)\big\} &\lesssim \beta \cdot D^2_{\beta} \sum_{v \in \ball(u,2R)}  \frac{d(u,v)}{R} \\
    &\lesssim \frac{\sqrt{\alpha}}{\ln (\nicefrac{1}{\alpha})} \cdot (\ln^2 (\nicefrac{1}{\alpha})) \sum_{v \in \ball(u,2R)} \frac{d(u,v)}{R} \\
    &\lesssim \frac{\sqrt{\alpha}}{\ln (\nicefrac{1}{\alpha})} \cdot \ln^2 (\nicefrac{1}{\alpha}) \cdot \frac{y_u}{\alpha} = A_{\infty} \cdot y_u,
\end{align*}
where the last inequality follows from the fact that each positive edge weight is at least $\alpha$. Thus, from~(\ref{eqn:disagreement_parsed}) it follows:
$$\disagree_u(\calP, E^+, E^-)\lesssim \ainf \cdot y_u.$$

We now prove Theorem \ref{thm:local_guarantee}, part (b). We separately consider short and long positive edges.
Let $E^{\leq R}$ be the set of positive edges $(v,w)\in E^+$ with $x_{vw}\leq R$. Note that 
\begin{align} \label{eq:y-bound}
y_u &\geq \sum_{v\in\ball(u,R)} w_{uv}\min(d(u,v), 1 - d(u,v))\\ &=\sum_{v\in\ball(u,R)} w_{uv} d(u,v) =\frac{1}{3}\notag
\sum_{v\in\ball(u,R)} w_{uv}\frac{d(u,v)}{R}.
\end{align}
Therefore, we have
\begin{align*}
\E[\disagree_u(\calP, E^{\leq R}, \varnothing) -  3D_\beta \cdot y_u] &\leq
 \E\Big[ \sum_{v \in \ball(u, R)} w_{uv} \cdot \one\{\calP(u)\neq \calP(v)\} - D_{\beta} \sum_{v \in \ball(u, R)} w_{uv} \frac{d(u,v)}{R} \Big]\\
&= \sum_{v \in \ball(u, R)} w_{uv} \Big(\Pr\{\calP(u) \neq \calP(v)\} - D_\beta \frac{d(u,v)}{R}\Big)\\
&\leq \sum_{v \in \ball(u, R)} w_{uv} \Big(\Pr\{\calP(u) \neq \calP(v)\} - D_\beta \frac{d(u,v)}{R}\Big)^+.
\end{align*}
%\knote{we should sum over $v\in \ball(u,R)$ with $(u,v)\in E^+$}
%\jnote{Don't get the first inequality. Where do we use Theorem 3.1 part (b) above?}\ynote{Jafar, I fixed this. Please check.}\jnote{Looks good. Also I am changing $y(u)$ to $y_u$ in this section.}
Since all edges $(u,v)$ in $E^{\leq R}$ are positive, we have $w_{uv}\leq 1$. Consequently,
\begin{align*}
    \E[&\disagree_u(\calP, E^{\leq R}, \varnothing) -  3D_\beta \cdot y_u] 
\leq\sum_{\stackrel{v \in \ball(u, R)}{ \text{s.t. }(u,v) \in E^+}} \Big(\Pr\{\calP(u) \neq \calP(v)\} - D_\beta \frac{d(u,v)}{R}\Big)^+.
\end{align*}
We bound the right hand side using property (2) of Theorem~\ref{thm: decompositon}:
%\begingroup
%\allowdisplaybreaks
\begin{align*}
\sum_{v \in \ball(u, R)} \Big(\Pr\{\calP(u) \neq \calP(v)\}  - D_\beta \frac{d(u,v)}{R}\Big)^+
&\lesssim \beta^2 \sum_{v \in \ball(u, 2R)} \frac{d(u,v)}{R}\\
&\lesssim \frac{\alpha}{\ln^2 (\nicefrac{1}{\alpha})} \sum_{v \in \ball(u, 2R)} d(u,v) \\
&\leq \frac{1}{\ln^2 (\nicefrac{1}{\alpha})}
\sum_{v \in \ball(u, 2R)} 
w_{uv}\cdot 2\min(d(u,v), 1 - d(u,v))\\
&\leq   \frac{2}{\ln^2 (\nicefrac{1}{\alpha})} \cdot y_u.
\end{align*}
%\endgroup
 Here, we used that $w_{uv} \geq \alpha$
 and $d(u,v) \leq 2(1 - d(u,v))$ for $v\in\ball(u, 2R)$. Thus,
$$
\E[\disagree_u(\calP, E^{\leq R}, \varnothing)] \lesssim \big(\ln (\nicefrac{1}{\alpha}) + \frac{1}{\ln^2 (\nicefrac{1}{\alpha})}\big) y_u \lesssim A_1 \cdot y_u.    
$$
Furthermore, $\disagree_u(\calP, E^+\setminus E^{\leq R}, \varnothing) \leq \frac{1}{R} \cdot y_u \leq A_1 \cdot y_u.$ Therefore, from~(\ref{eqn:disagreement_parsed}) it follows that 
$$
    \E[\disagree_u(\calP, E^+, E^-)]\lesssim A_1 y_u.
$$
%Thus,
%$$
%\E[\disagree_u(\calP, E^{\leq R}, \varnothing)] \leq \big(\log (\nicefrac{1}{\alpha}) + %\frac{1}{\log^2 (\nicefrac{1}{\alpha})}\big) y_u \lesssim A_1 \cdot y_u.    
%$$
%
%For the set of  edges $E^{> R} = E^+\setminus E^{\leq R}$, we have 
%$$\disagree_u(\calP, E^{> R}, \varnothing) \leq \frac{1}{R} \cdot y_u \leq A_1 \cdot y_u.$$
%We conclude that
%\begin{align*}
%\E&[\disagree_u(\calP, E^+, E^-)] 
%=  \E[\disagree_u(\calP, E^{\leq  R}, \varnothing)] \\
%&  + \E[\disagree_u(\calP, E^{> R}, \varnothing)] + \E[\disagree_u(\calP, \varnothing, %E^-)] \lesssim A_1 \cdot y_u.
%\end{align*}
\end{proof}
We now use Theorem \ref{thm:local_guarantee} to prove Theorem~\ref{thm:main}.

\subsection{Proof of Theorem~\ref{thm:main}}\label{subsec:prfMain}
In this section, we show that the partitioning $\calP$ output by Algorithm~\ref{alg:CC} achieves the desired approximation ratio -- thereby proving our main theorem, Theorem~\ref{thm:main}.
To show this, we will use the fact that $\calP$ satisfies the properties in Theorem~\ref{thm:local_guarantee}.

\begin{proof}[Proof of Theorem~\ref{thm:main}.] If $p=\infty$, then we get an $O(A_{\infty}) = O((\nicefrac{1}{\alpha})^{\nicefrac{1}{2}} \ln\nicefrac{1}{\alpha})$ approximation by Theorem~\ref{thm:local_guarantee}, item (a), as desired. So we assume that $p < \infty$ below.  Given the guarantees from Theorem \ref{thm:local_guarantee}, we observe, 
\begingroup
\allowdisplaybreaks
\begin{align*}
    \E\left[\sum\limits_{u\in V}\disagree_u(\calP, E^+, E^-)^p \right] 
    &=\sum\limits_{u\in V}\E[\disagree_u(\calP, E^+, E^-)^{p-1}  \cdot \disagree_u(\calP, E^+, E^-) ]\\
    &\lesssim \sum\limits_{u\in V}\E\left[(\ainf\cdot y_u)^{p-1}\cdot \disagree_u(\calP, E^+, E^-)\right]\\
    &=\sum\limits_{u\in V}(\ainf\cdot y_u)^{p-1}\E\left[\disagree_u(\calP, E^+, E^-)\right]\\
    &\lesssim \sum\limits_{u\in V}(\ainf\cdot y_u)^{p-1}\cdot A_1\cdot y_u=\sum\limits_{u\in V}A^p\cdot y_u^{p},
\end{align*}
\endgroup
where $A=\left(\ainf^{p-1}\cdot A_1\right)^{\frac{1}{p}}$. Note that the desired approximation factor is $O(A)$. From Jensen's inequality, it follows that
\begingroup
\allowdisplaybreaks
\begin{align*}
    \E\bigg[\Big(\sum\limits_{u\in V}\disagree_u(\calP,  E^+, E^-)^{p}\Big)^{\frac{1}{p}}\bigg]
    &\leq \left(\E\left[\sum\limits_{u\in V}\disagree_u(\calP, E^+, E^-)^p\right]\right)^{\frac{1}{p}}\\
    &\lesssim \left(\sum\limits_{u\in V}A^p\cdot y_u^p\right)^{\frac{1}{p}}= A\cdot \Vert y\Vert_p.
\end{align*}
\endgroup
This finishes the proof.
\end{proof}

%\snote{To change LP to CP in proof of Theorem 4.1}
%\knote{Do we use letter $X$ or $Y$ for the metric space?}
% \jnote{Changed LP to CP and Y to X in this section.}
%\knote{Please, replace $q$ with $p$.}
% In this section, we present an approximation algorithm for minimizing local objectives for Correlation Clustering with Asymmetric Classification Errors. This will prove our main result - Theorem~\ref{thm:main}. The algorithm first solves the convex programming relaxation (Figure~\ref{fig:CP}), and assigns every edge a length of $x_{uv}$.
% %The algorithm first solves the convex programming relaxation (Figure~\ref{fig:CP}), and obtains an optimal fractional solution $\{x_{uv}\}_{u,v \in V}$.
% It then considers a metric space $(X = V,d)$ where $d(u,v)= d(v,u) = x_{uv}$ for all $u,v,\in V$.
% %It then defines a metric space $(Y = V,d)$ over the data to be clustered. The distances $d(\cdot, \cdot)$ are set according to the convex programming solution $x$, that is, for every $u,v \in V$, we have $d(u,v) = d(v,u) = x_{uv}$.
% Finally, the algorithm uses our technical result, Theorem~\ref{thm: decompositon}
% to obtain a partitioning $\calP$, where the parameters are set as follows: $r = \frac{\sqrt{\alpha}}{\log (\nicefrac{1}{\alpha})}$ and $ R = \nicefrac{1}{3}$\footnote{Thus, $\beta = \nicefrac{r}{R} =  \frac{3\sqrt{\alpha}}{\log (\nicefrac{1}{\alpha})}$ and $D_\beta =  O(\log\nicefrac{1}{\beta})) = O(\log\nicefrac{1}{\alpha})$)}. The partitioning $\calP$ is the solution that is output by our algorithm. We give a pseudo-code for this algorithm in Algorithm~\ref{alg:CC}.

\section{Overview of Metric Partitioning Scheme}\label{sec:proof_overview}
In this section we describe our partitioning scheme and give a proof overview of Theorem~\ref{thm: decompositon}. A pseudocode for this partitioning scheme is given in Algorithm~\ref{alg:MPSP}.
More specifically, in Section~\ref{sec:iterative_clust} we reduce the problem to choosing a random set of particular interest as stated in Theorem~\ref{thm:itr-clustering}. In Section~\ref{sec:single_clust} we describe an algorithm for choosing such a random set and give a proof overview of its correctness. The pseudocode for choosing a random set is given in Algorithm~\ref{alg:ball-select}.
\subsection{Iterative Clustering}\label{sec:iterative_clust}
Given a metric space $(X,d)$, our partitioning scheme uses an iterative algorithm -- Algorithm~\ref{alg:MPSP} to obtain $\calP$. Let $X_t$ denote the set of not-yet clustered vertices at the start of iteration $t$ of Algorithm~\ref{alg:MPSP}.
At step $t$, the algorithm finds and outputs random set $P_t \subseteq X_t$. It then updates the set of not-yet clustered vertices $(X_{t+1} = X_t \setminus P_t)$, and repeats this step until all vertices are clustered.
% We summarize this iterative process in Algorithm~\ref{alg:MPSP}.
Algorithm~\ref{alg:MPSP} makes use of the following theorem in each iteration to find the random set $P_t$.

We need the following notation to state the theorem.
Let $\delta_P(u,v)$ be the cut metric induced by the set $P$: $\delta_P(u,v)=1$ if $u\in P$ and $v\notin P$ or 
$u\notin P$ and $v\in P$; $\delta_P(u,v)=0$ 
if $u\in P$ and $v\in P$ or $u\notin P$ and $v\notin P$.
Also, let $\vee_P(u,v)$ be the indicator of the event $u\in P$ or $v\in P$ or both $u$ and $v$ are in $P$. We denote $[k]= \{1, 2, \ldots k\}$.
\begin{theorem}\label{thm:itr-clustering}
For every $q \geq 1$ there exists a $\beta^*_q=\Theta\big(\frac{1}{q\ln(q+1)}\big) < 1$ such that the following holds. Consider a finite metric space $(X,d)$. 
Fix two positive numbers $r$ and $R$ such that $\beta = r/R\leq \beta^*_q$. 
Let $D_{\beta}=2(q+1)\ln\nicefrac{1}{\beta}$. Then, there exists an algorithm for finding a random set $P$ satisfying properties (a), (b), and (c):

\item[(a)] $\diam(P)\leq 2R$ (always);

\item[(b)] For every point $u$ in $X$, the following bound holds:
$$\sum_{\mathclap{\substack{\phantom{1}\\v\in \ball(u,R)}}}\Big(\Pr\big\{\delta_P(u,v) =1 \big\}  - D_{\beta}\frac{d(u,v)}{R}\,\Pr\{\vee_P(u,v)=1\}\Big)^+ \lesssim
\beta^q \cdot 
\E\Bigg[\sum_{v\in \ball(u,2R)} \frac{d(u,v)}{R}
\cdot \vee_P(u,v)\Bigg].
$$
\item[(c)] Moreover, for every $u$ in $X$, we always have 
$$
\sum_{v\in \ball(u,r)}\delta_P(u,v)
\lesssim \beta \cdot D^2_{\beta} \cdot
\sum_{v\in\ball(u,2R)}\frac{d(u,v)}{R}
\cdot \vee_P(u,v).
$$
\end{theorem}
Informally, Theorem~\ref{thm:itr-clustering} is a ``single-cluster'' version of Theorem~\ref{thm: decompositon}, and there is a one-to-one correspondence between their properties.
In Section~\ref{sec:itr-clustering}, we show that Theorem~\ref{thm: decompositon} holds for $\calP$ if we assume that each partition $P \in \calP$ satisfies Theorem~\ref{thm:itr-clustering}. Thus, 
to obtain Theorem~\ref{thm: decompositon}, it remains to prove
Theorem~\ref{thm:itr-clustering}.

\subsection{Selecting a Single Cluster}\label{sec:single_clust}
We will use the following definitions. Let $r$ and $R$ be positive numbers with $r< R$. Define $\beta = {r}/{R}\leq \beta^*_q$ and $D_\beta = 2(q+1)\ln \nicefrac{1}{\beta}$ where $q\geq 1$. 
Let $R_0 = \nicefrac{R}{D_\beta}$ and $R_1 = R - R_0$. We let $\rho_q(\beta)=(\nicefrac{1}{\beta})^{q+1}$ (see Figure~\ref{fig:lightball-radii}). We choose $\beta^*_q$ so that $r < R_0 < R$ (see Section~\ref{sec:Cluster-select} for details).
\begin{figure}
    \centering
    \scalebox{0.75}{\resizebox {\columnwidth/2} {!} {
\begin{tikzpicture}

\fill[fill=white,draw=black,thick] (0,0) ellipse (5);
% \fill[fill=white,draw=black,thick, dotted] (0,0) ellipse (4.8);
%\fill[fill=black!20!white,draw=black,dotted] (0,0) ellipse (4.2);
\fill[fill=white,draw=black,thick] (0,0) ellipse (4.2);

%\fill[fill=black!40!white,draw=black!20!white] (0,0) ellipse (2.7);
%\fill[fill=white,draw=white] (0,0) ellipse (2.3);

\draw[draw=black,thick] (0,0) ellipse (2.5);
\fill[fill=white,draw=black,thick] (0,0) ellipse (1);

\node (z) at (0,0) {\huge $z$};

\node (v1) at (1.2,0) {};
 \node (v2) at (-3.1,3.1) {};
\node (v3) at (-3.7,-3.7) {};
\node (v4) at (2.8, -2.8) {};
\node (v5) at (3.7, -3.7) {};
\node (v6) at (6.4, -2.5) {\huge$R_0 = \nicefrac{R}{D_\beta}$};
\node (v9) at (-5.9, -4.4) {\huge$R$};
\node (v10) at (-4.9, 4.9) {\huge$R_1$};
\node (v12) at (1.1, 2.4) {};
\node (v13) at (0.5, 1.1) {};

\node (v11) at (4.9, 4.9) {\huge$t$};

\draw[latex'-latex'] (z) -- node[above] {} (v1);
\draw[latex'-latex'] (z) -- node[above] {} (v2);
\draw[latex'-latex'] (z) -- node[above] {} (v3);
\draw[latex'-latex'] (v4) -- node[above] {} (v5);
\draw[latex'-latex'] (z) -- node[above] {} (v12);

\node (v7) at (3.25, -3.25) {};
\node (v8) at (0.5, -0.1) {};
\node (u1) at (-2.1, -2.2) {};
\node (u2) at (-2.1, 2.2) {};
\draw [black,dotted,thick]   (v6) to[out=200,in=70] (v7);
\draw [black,dotted,thick]   (v6) to[out=160,in=-30] (v8);
\draw [black,dotted,thick]   (v9) to[out=15,in=160] (u1);
\draw [black,dotted,thick]   (v10) to[out=-90,in=180] (u2);
\draw [black,dotted,thick]   (v13)  to[out=90,in=180] (v11);

\node (u3) at (-4.8,0) {};
\node (u4) at (-5,0) {};
%%\draw[latex'-latex'] (u3) -- node[above] {} (u4);
%%\node (u5) at (-5.9, 0) {\huge$r$};
%%\draw [black,dotted,thick]   (u5) to[out=-90,in=-90] (u3);

%\node (v14) at (2.4, 1) {};
%\node (v15) at (7.2, 1.2) {\huge$\gamma-$light shell};
%\node (v16) at (7.8, 0.4) {\huge of width $r$};
%\draw [black,dotted,thick]   (v14)  to[out=60,in=180] (v15);

% \node[shape=circle,fill=black] (j) at (-10,0) {};
% \node[text=black] () at (-10,-0.5) {$U$};
% \definecolor{pink}{RGB}{255,179,186}
% \definecolor{blue1}{RGB}{186,225,255}
% \node[draw, fit={(-10.5,2) (-9.5,3)}, inner sep=0pt, label=center:U, fill=pink] (i1) {};
% \node[draw, fit={(-10.5,-2) (-9.5,-3)}, inner sep=0pt, label=center:1, fill=blue1] (i2) {};
% \node[text=black] () at (-10,-4.5) {$1$};

\end{tikzpicture}}

% \begin{tikzpicture}

% \fill[fill=white,draw=black,thick] (0,0) ellipse (5);
% % \fill[fill=white,draw=black,thick] (0,0) ellipse (3);
% \fill[fill=black!20!white,draw=black,dotted] (0,0) ellipse (4.2);
% \fill[fill=white,draw=black,thick] (0,0) ellipse (1);
% \node (z) at (0,0) {\huge $z$};

% \node (v1) at (1.2,0) {};
%  \node (v2) at (-3.1,3.1) {};
% \node (v3) at (-3.7,-3.7) {};
% \node (v4) at (2.8, -2.8) {};
% \node (v5) at (3.7, -3.7) {};
% \node (v6) at (4.9, 4.9) {\huge$R_0$};
% \node (v9) at (4.9, 4.9) {\huge$R$};
% \node (v10) at (4.9, 4.9) {\huge$R_1$};

% \draw[latex'-latex'] (z) -- node[above] {} (v1);
% \draw[latex'-latex'] (z) -- node[above] {\huge $R_1$} (v2);
% \draw[latex'-latex'] (z) -- node[above] {\huge$R$} (v3);
% \draw[latex'-latex'] (v4) -- node[above] {} (v5);
% \draw[latex'-latex'] (v4) -- node[above] {} (v5);

% \node (v7) at (3.25, -3.25) {};
% \node (v8) at (0.5, 0.3) {};
% \draw [black,dotted,thick]   (v6) to[out=-20,in=70] (v7);
% \draw [black,dotted,thick]   (v6) to[out=160,in=70] (v8);

% % \node[shape=circle,fill=black] (j) at (-10,0) {};
}
    \caption{Balls with Different Radii}\label{fig:lightball-radii}
    $R > r > 0$, $q\geq 1$, $\beta =r/R$, $D_\beta = 2(q+1) \ln\nicefrac{1}{\beta}$,
    %$r =\nicefrac{\sqrt{\alpha}}{\log (\nicefrac{1}{\alpha})}, R > r$, $\beta =r/R$, $D_\beta = \log \nicefrac{1}{\beta}$,
    $R_0 = R/D_{\beta}$, $R_1 = R - R_0$.
\end{figure}

Given a metric space $(X,d)$ and parameters $r$ and $R$, our procedure for finding a random set $P \subseteq X$ begins by finding a pivot point $z$ with a densely populated neighborhood -- namely, $z$ is chosen such that a ball of radius $R_0$ around $z$ contains the maximum number of points. More formally,
\begin{align}\label{def:heavy_core}
    z=\arg \max_{u\in X}|\ball(u,R_0)|.
\end{align}
We refer to this ball of small radius around $z$ as the ``core'' of the cluster. 
Our choice of the pivot $z$ is inspired by the papers by~\citet{CGW03,PueMil18,CharikarGS17}.
We then consider a ball of large radius $R_1$  around the pivot $z$ and examine the following two cases -- ``Heavy Ball'' and ``Light Ball''.
If this ball of large radius around $z$ is sufficiently populated, that is, if the number of points in $\ball(z,R_1)$ is at least  $(\nicefrac{1}{\beta})^{q+1}$ times the number of points in the core, we call this case ``Heavy Ball''. In the case of Heavy Ball, we will show that $P = \ball(z,R_1)$ (a ball around $z$ of radius slightly less than $R$) satisfies the properties of Theorem~\ref{thm:itr-clustering}.
In the case of ``Light Ball'', the ball of large radius around $z$ is not sufficiently populated. In this case, the algorithm finds a radius $t$ $(t \leq R)$ such that $P = \ball(z,t)$ satisfies the properties of Theorem~\ref{thm:itr-clustering}.    
In the following subsections we provide an overview of the proof for these two cases. A formal proof of Theorem~\ref{thm:itr-clustering} can be found in Section~\ref{sec:Cluster-select}. 

\begin{algorithm}[tb]
   \caption{Metric Space Partitioning Scheme}
   \label{alg:MPSP}
\begin{algorithmic}
   \STATE {\bfseries Input:} Metric Space $(X,d)$ and $r, R > 0$, $q\geq 1$.
   \STATE Define $t=0$, $X_1 = X$.
   \REPEAT
   \STATE $t = t+1$.
   \STATE $P_t =$ Cluster Select$(X_t, d, r, R, q)$.
   \STATE $X_{t+1} = X_t \setminus P_t$.
   \UNTIL $X_t = \varnothing$
   \STATE Output $(P_1, P_2, \ldots, P_t)$.
\end{algorithmic}
\end{algorithm}
\begin{algorithm}[t]
   \caption{Cluster Select}
   \label{alg:ball-select}
\begin{algorithmic}
   \STATE {\bfseries Input:} Metric space $(X,d)$ and $r, R > 0$, $q\geq 1$ \\
   \STATE Define: $\beta = \nicefrac{r}{R}, D_\beta = 2(q+1)\ln \nicefrac{1}{\beta}$ .\\
   \STATE Define: $R_0 = \nicefrac{R}{D_\beta}, R_1 = R - R_0, \rho_q(\beta) = (\nicefrac{1}{\beta})^{q+1}$.\\
   \STATE Select  $z=\arg \max_{u\in X}|\ball(u,R_0)|.$\\
   \IF{$|\ball(z, R_1)| \geq \rho_q(\beta) \cdot |\ball(z, R_0)|$}
   \STATE Set $P = \ball(z, R_1)$.
   \ELSE 
   \STATE Consider $S$ as stated in Definition~\ref{def:set_S}.
   \STATE Consider $\pi^{inv}_S$ as stated in Definition~\ref{def:function-pi}.
   \STATE Let $F$ be the cumulative distribution function stated in Definition~\ref{def:F}.
   \STATE Choose a random $x\in [0,R/2]$ according to $F$.
   \STATE Set $P = \ball(z, \pi^{inv}_S(x))$.
   %\STATE Define $\pi_S$ according to Definition~\ref{def:function-pi}.
   %\STATE Choose $t \in [0, \pi(S)]$. Set $P = \ball(z, t)$.
%   \STATE\jnote{Changed the pseudocode. Neither here nor in Algorithm 2 we set the value for $q$.}\ynote{I think $q$ should be part of the input in Algorithms 2 and 3; in Algorithm 1, I now assign $q=2$}\jnote{Added q in algo 2 and algo 3}
   \ENDIF
   \STATE Output $P$.
\end{algorithmic}
\end{algorithm}

\subsubsection{Heavy Ball}
The Heavy Ball $P$ is a ball of radius $R_1$ around $z$ which contains many points. As the diameter of $P$ is $2R_1 < 2R$, it is easy to see that a Heavy Ball satisfies property $(a)$ of Theorem~\ref{thm:itr-clustering}. We now focus on showing that properties $(b)$ and $(c)$ hold for Heavy Ball.
%\knote{we can be more specific: it's radius is $\cdots$, thus it's diameter is at most $\cdots$. Hence, it satisfies property (a)}. 
% To show that a Heavy Ball also satisfies properties $(b)$ and $(c)$, we will use the large number of points in $P$ to our benefit\knote{this sentence needs work}. 
Observe 
as $z$ was chosen according to~(\ref{def:heavy_core}), for every point $u \in X\setminus \{z\}$, $u$ has a less populated neighborhood of radius $R_0$ than that of $z$. 
This combined with the fact that $\ball(z, R_1)$ is heavy, implies that for every $u$, there are sufficiently many points in $P$ at a distance of at least $R_0$ from $u$.
% This implies that for every point $u \in P\setminus \{z\}$, $u$ is not surrounded by very many points and hence that the points in $P$ are well spread.\jnote{Can we restructure the above two sentences? Currently, it looks like there is exactly one such $z$. Also we use ``points" multiple times in the sentence.} 
Thus, for any point $u \in X$, we can expect the sum of distances between $u$ and the points in $P$ to be large. In fact, we show that the left hand sides of properties $(b)$ and $(c)$ can be charged to $\sum_{v \in P} \frac{d(u,v)}{R}$, the sum of distances between $u$ and the points in $P$. For points $u$ such that $d(z,u) \leq R$, $P \subseteq \ball(u, 2R)$ and hence, $\sum_{v \in \ball(u,2R)} \frac{d(u,v)}{R} \vee_P(u,v) \geq \sum_{v \in P} \frac{d(u,v)}{R}$. Thus, for every $u \in X$, we can charge the left hand sides of properties $(b)$ and $(c)$ to the quantity $\sum_{v \in \ball(u,2R)} \frac{d(u,v)}{R} \vee_P(u,v)$.
This allows us to conclude that a Heavy Ball satisfies Theorem~\ref{thm:itr-clustering}.
% \jnote{The last sentence is confusing.} \snote{tried to fix it. Does it look better?}\jnote{Not confusing anymore.}

\subsubsection{Light Ball}\label{sec:light_ball_main}
%\jnote{My suggestions: I think we should avoid using cut-edges since we never mention complete graph interpretation of a metric space. Also we should formally define the set S and use inequality~\ref{eq:newC} in its definition. We should describe what S looks like (finite union of intervals) and explain why it is easy to find it. It might be better to explain the connection among a light ball, non-steady growth and $S$ having large measure by using light shells; recall that in the proof we show that $S_\gamma \subset S$ and $\mu(S_\gamma)\geq R/2$. I think we should also bring the definitions of $\pi_S,\pi_S^{inv}, F$ here since we refer to them in our pseudocode. We should elaborate more on how the choice of $F$ is beneficial in addition to referring to our previous work.}

In this subsection, we consider the case of $|\ball(z, R_1)| < \rho_q(\beta) \cdot |\ball(z, R_0)|$, which we call Light Ball. In the case of Light Ball, we choose a random radius $t\in (0,R_1]$ and set $P=\ball(z,t)$. Observe that property (a) of Theorem~\ref{thm:itr-clustering} holds trivially since the radius $t < R$.

Now consider property~$(c)$ of Theorem~\ref{thm:itr-clustering}. Recall that for every point $u\in X$, property~$(c)$ gives a bound on the total number of points separated from $u$ (by $P$) residing in a small ball $\ball(u,r)$, i.e., $\sum\limits_{v\in\ball(u,r)}\delta_P(u,v)$. Note that property (c) gives a deterministic guarantee on $P$.
Therefore, we choose a random radius $t\in(0,R_1]$ from the set of all radii for which property (c) of Theorem~\ref{thm:itr-clustering} holds. More specifically, we define the following set.
\begin{definition}\label{def:set_S}
Let $S$ be the set of all radii $s$ in $(3R_0,R_1]$ such that for every $u\in X$ set $P=\ball(z,s)$ satisfies:
\begin{align}\label{eq:newC}
\sum_{v\in \ball(u,r)}\delta_P(u,v)
\leq
25
\beta \cdot D^2_{\beta} \cdot
\sum_{v\in\ball(u,2R)}\frac{d(u,v)}{R}
\cdot \vee_P(u,v).
\end{align}
\end{definition}

The set $S$ can be computed in polynomial time since the number of distinct clusters $P=\ball(z,t)$ is upper bounded by the size of the metric space, $\vert X \vert$. By the same token, $S$ is a finite union of disjoint intervals.
%In the case of Light Ball, the number of points in $\ball(z,R)$ is not enough to guarantee that properties $(b)$ and $(c)$ hold for $P = \ball(z,R)$.
%In the case of Light Ball, the number of points in a ball of radius $R$ around $z$ presents an initial challenge.
%For every $u\in X$, we cannot guarantee that the left hand sides of properties $(b)$ and $(c)$ can be charged to the number of points in $\ball(z,R)$.
%However, we will show that we can choose a radius $t \in (0,R]$ such that $P = \ball(z,t)$ satisfies Theorem~\ref{thm:itr-clustering}.

Now we show why we can expect the set $S$ to be large. Consider $P=\ball(z,s)$ such that $s\in S$. 
As $S$ is computed according to Definition~\ref{def:set_S}, it implies that the boundary of $P$ is somewhat sparsely populated -- as for every $u \in X$, it bounds the number of points within a small neighborhood of $\ball(u,r)$ that are separated from $u$ (note that $\sum_{v \in \ball(u, r)} \delta_P(u,v)$ is trivially $0$ for points $u$ that are not close to the boundary of $P$). 
%We start by considering property $(c)$. Let $S$ be the set of all radii $s \in (0,R]$ such that $\ball(z,s)$ satisfies property $(c)$. We will show why we can expect the set $S$ to be non-empty. If for a certain $s \in (0,R]$, $P = \ball(z,s)$ satisfies property $(c)$, it implies that the boundary of $P$ is somewhat sparsely populated -- as for every $u \in X$, the number of points within $\ball(u,r)$ that are separated from $u$, that is, $\sum_{v \in \ball(u, r)} \delta_P(u,v)$ is at most $O(\beta \cdot D^2_{\beta} \cdot \sum_{v\in\ball(u,2R)}\frac{d(u,v)}{R} \cdot \vee_P(u,v))$ (note that $\sum_{v \in \ball(u, r)} \delta_P(u,v)$ is trivially $0$ for points $u$ that are not close to the boundary of $P$).
Since $\ball(z,R_1)$ does not contain many points, the number of points in $\ball(z,s')$ cannot grow too quickly as we increase the radius $s'$ from $0$ to $R_1$. This suggests that for many of such radii $s'$, the ball $P = \ball(z, s')$ has a sparsely populated boundary, and hence the set $S$ should be large. In fact, we use the above argument to show that the Lebesgue measure of the set $S$ satisfies $\mu(S) \geq \nicefrac{R}{2}$. This will allow us to define a continuous probability distribution on $S$.
%This suggests that there should be many points $s$ such that $P = \ball(z,s)$ satisfies property $(c)$. In fact, we use the above argument to show that the Lebesgue-measure of the set $S$, $\mu(S) \geq \nicefrac{R}{2}$, in other words, for a majority fraction of the radii $s \in (0,R]$, $P = \ball(z,s)$ satisfies property $(c)$. 
% On the other hand, we use the fact that $\ball(z,R)$ does not contain many points to argue that as we increase the radius $t$ from $0$ to $R$, the number of points in $\ball(z, t)$ cannot grow too quickly. because if the number of points in the ball grew steadily as we increased the radius from $0$ to $R$, it would imply that $\ball(z,R)$ is heavily populated. More specifically, consider property $(c)$ of Theorem~\ref{thm:itr-clustering}. If a cluster $P$ satisfies property $(c)$ of Theorem~\ref{thm:itr-clustering}, it means that the cluster is sparsely populated around the boundary - as for every $u \in X$, the total number of cut-edges of length at most $r$ that are incident on $u$ is upper-bounded by $\beta \cdot D^2_{\beta} \cdot \sum_{v\in\ball(u,2R)}\frac{d(u,v)}{R} \cdot \vee_P(u,v)$. We show that there are many radii $t \in (0,R]$ such that $P = \ball(z,t)$ satisfies property $(c)$ of Theorem~\ref{thm:itr-clustering}. More formally, let $S \subset (0,R)$ be such that for every $t \in S$, $\ball(z,t)$ satisfies property $(c)$ of Theorem~\ref{thm:itr-clustering}. We show that the Lebesgue-measure of the set $S$ is at least $\nicefrac{R}{2}$, that is, $\mu(S) \geq \nicefrac{R}{2}$. 

What remains to be shown is that for a random radius $t \in S$, the set $P=\ball(z,t)$ satisfies property $(b)$ of Theorem~\ref{thm:itr-clustering}. For this purpose we define a measure preserving transformation $\pi_S$ that maps an arbitrary measurable set $S$ to the interval $[0, \mu(S)]$.
\begin{definition}\label{def:function-pi}
Consider a measurable set $S\subset [0,R]$. Define function 
$\pi_S:[0,R]\to [0,\mu(S)]$ as follows
$\pi_S(x) = \mu([0,x]\cap S)$. Also, for $y\in [0,\mu(S)]$,
let 
$$\pi_S^{inv}(y) = \min\{x: \pi_S(x) = y\}.$$
\end{definition}
Recall that the set $S$ stated in Definition~\ref{def:set_S} is a finite union of disjoint intervals. In this case, what $\pi_S$ does is simply pushing the intervals in $S$ towards $0$, and thus, allowing us to treat the set $S$ as a single interval $[0, \mu(S)]$. For the rest of the proof overview, we assume that $S=[0, \mu(S)]$ and $\pi_S$ is the identity. This simplifies the further analysis of Theorem~\ref{thm:itr-clustering} immensely.
%Now, if we can find a radius $t \in S$ such that $\ball(z,t)$ satisfies property $(b)$ of Theorem~\ref{thm:itr-clustering}, we will be done. As $S$ is possibly a union of finitely-many disjoint intervals, we start by defining a measure preserving transformation $\pi_S$ that maps the set $S$ to the interval $[0, \mu(S)]$. The transformation $\pi_S$ simply pushes the intervals in $S$ towards the cluster center $z$, and thus, it allows us to assume that the set $S$ lies in the continuous interval $[0, \mu(S)]$.

Next, we define a cumulative distribution function $F$ on $[0, \nicefrac{R}{2}] \subseteq [0, \mu(S)]$:
\begin{definition}\label{def:F}
Let $F:\;[0, R/2]\to [0,1]$ be a cumulative distribution function such that
\begin{align}\label{func_L1}
F(x) = 
\frac{1-e^{-\nicefrac{x}{R_0}}}{1-e^{-\nicefrac{R}{2R_0}}}.
\end{align}
\end{definition}
%Next, we define an appropriate probability distribution function $F$ on the interval $[0, \nicefrac{R}{2}] \subseteq [0, \mu(S)]$, and choose a random point $x \in [0, \nicefrac{R}{2}]$ according to the distribution $F$.

We choose a random $x\in [0,R/2]$ according to $F$ and set $P=\ball(z,\pi_S^{inv}(x))$ (see Algorithm~\ref{alg:ball-select}). Since we assume in this proof overview that $\pi_S$ is the identity, $P=\ball(z,x)$.
%Recall that one of our objectives is to choose a radius $t$ from $S$ so that $P=\ball(z,t)$ satisfies property (c) of Theorem~\ref{thm:itr-clustering}.
Now, we show that the radius $x$ chosen in such a manner guarantees that the cluster $P$ satisfies property $(b)$. Loosely speaking, the 
motivation behind our particular choice of cumulative distribution function $F$ is the following: For two points $u,v \in X$, function $F$ bounds the probability of $u$ and $v$ being separated by $P$, in terms of $D_\beta$ times the probability that either $u$ or $v$ lies in $P$. Unfortunately, this bound does not hold for points $u$ with $d(z,u)$ close to $\nicefrac{R}{2}$. However, the choice of parameters for function $F$ in Definition~\ref{def:F} gives us two desired properties.
Without loss of generality assume that $d(z,u) \leq d(z,v)$. Then, the probability that $P$ separates the points $u$ and $v$, $\Pr(\delta_P(u,v)) = \Pr(d(z,u) \leq x \leq d(z,v)) = F(d(z,v)) - F(d(z,u))$. Moreover, as $d(z,u) \leq d(z,v)$, the probability that either $u$ or $v$ lies in $P$, $\Pr(\vee_P(u,v)) = 1 - F(d(z,u))$. 
Thus, choosing $F$ according to Definition~\ref{def:F} ensures:
\begin{itemize}
    \item (Property I) $F(d(z,v)) - F(d(z,u))$ is bounded in terms of $D_\beta$ times $1 - F(d(z,u))$ (Please see Claim~\ref{cl:ineq-on-F} for a formal argument).\label{Prop:F_1}
    \item (Property II) The probability that the cluster $P$ includes points $u$ such that $d(z,u) > \nicefrac{R}{2} - R_0$, is very small (please see Claim~\ref{cl:apx:tailF}).\label{Prop:F_2}
\end{itemize}
In fact, (Property II) of function $F$ is the reason why we are able to guarantee that property $(b)$ satisfies~(\ref{eq:separating-condition}) only on average, with the error term coming from our inability to guarantee~(\ref{eq:separating-condition}) for points on the boundary. We refer the reader to Section~\ref{sec:L_one} for a formal proof. Thus, we conclude the case of Light Ball and show that it satisfies Theorem~\ref{thm:itr-clustering}.

\section{Proof of Theorem~\ref{thm: decompositon}}\label{sec:itr-clustering}
% In this section, we present our technical result - an algorithm for partitioning a given metric space $(Y, d)$ into a number of clusters $\calP = (P_1, \dots, P_k)$ (where $k$ is not fixed).
% The algorithm iteratively computes the partitioning; in each iteration, it finds a cluster of vertices and then removes them from the set of unclustered vertices. Denote the set of unclustered vertices before the start of iteration $t$ by $Y_t$. Initially, $Y_1 = Y$ (no vertices are clustered). In iteration $t$, the algorithm finds a cluster $P_t \subset Y_t$
% and sets  $Y_{t+1} = Y_t \setminus P_t$. The algorithm stops when all vertices are clustered. It outputs the obtained partitioning $\calP = (P_1,\dots, P_k)$, where $k$ is the total number of iterations. This scheme is summarized in Algorithm~\ref{alg:MPSP}.

In this section, we present the proof of our main technical result -- Theorem~\ref{thm: decompositon} -- an algorithm for partitioning a given metric space $(X, d)$
into a number of clusters $\calP = (P_1, \dots, P_k)$ (where $k$ is not fixed).

Recall our iterative process for obtaining this partitioning -- Algorithm~\ref{alg:MPSP} -- which makes use of Theorem~\ref{thm:itr-clustering} in each iteration to select a cluster from
the set of not-yet clustered vertices.

The proof of Theorem~\ref{thm:itr-clustering} is presented in Section~\ref{sec:Cluster-select}.
We now present the proof of Theorem~\ref{thm: decompositon} assuming Theorem~\ref{thm:itr-clustering}.

\begin{proof}[Proof of Theorem~\ref{thm: decompositon}]
Property (a) of Theorem~\ref{thm:itr-clustering} guarantees that $\diam(P_i) \leq 2R$ for every $i\in [k]$ and thus property (1) of Theorem~\ref{thm: decompositon} holds.

We now show that property (2) holds. Fix $u\in X$. Consider iteration $i\in [k]$. Note that set $P_i$ satisfies property (b) of Theorem~\ref{thm:itr-clustering} regardless of what set $X_i$ we have in the beginning of iteration $i$. That is, for every set $Y\subset X$ and $u\in Y$, we have
\begin{multline}\label{eq:pi-part-b}
    \sum_{v\in \ball(u,R)\cap Y}\bigg(\Pr\big\{\delta_{P_i}(u,v) =1 \given X_i = Y\big\} - 
    D_{\beta}\frac{d(u,v)}{R}\,\Pr\{\vee_{P_i}(u,v)=1 \given X_i = Y\}\bigg)^+\\
    \lesssim    \beta^q \cdot 
    \E\Bigg[\sum_{v\in \ball(u,2R)\cap Y} \frac{d(u,v)}{R}
    \cdot \vee_{P_i}(u,v) \given X_i = Y\Bigg].
\end{multline}
We observe that inequality (\ref{eq:pi-part-b}) can be written as follows (for all $u\in X$).
\begingroup
\allowdisplaybreaks
\begin{multline}\label{eq:pi-part-b-prime}
    \sum_{v\in \ball(u,R)}\bigg(\Pr\big\{\delta_{P_i}(u,v) =1 \text{ and } u,v\in X_i\given X_i = Y\big\}\\ - 
    D_{\beta}\frac{d(u,v)}{R}\,\Pr\{\vee_{P_i}(u,v)=1 \text{ and } u,v\in X_i\given X_i=Y\}\bigg)^+\\
    \lesssim    \beta^q \cdot 
    \E\bigg[\sum_{v\in \ball(u,2R)} \frac{d(u,v)}{R}
    \cdot \vee_{P_i}(u,v)
    \cdot \one\left\{u,v\in X_i\right\} \given X_i= Y\bigg].
\end{multline}
\endgroup

If $u\notin Y$, then all terms in (\ref{eq:pi-part-b-prime}) are equal to 0, and the inequality 
trivially holds. If $u\in Y$, then corresponding terms in (\ref{eq:pi-part-b}) and (\ref{eq:pi-part-b-prime}) with $v\in Y$ are equal to each other; all terms in (\ref{eq:pi-part-b-prime}) with $v\notin Y$ are equal to 0. 
Denote the event that $u,v\in X_i$ by $\calE_{vi}$ (that is, $\calE_{vi}$ happens if both points $u$ and $v$ are not clustered at the beginning of iteration $i$).
We take the expectation of  (\ref{eq:pi-part-b-prime}) over $X_i=Y$ and add up the inequalities over all $i\in [k]$. Using the subaddivity of function $x\mapsto x^+$, we obtain
% \begingroup
% \allowdisplaybreaks
% \begin{multline}\label{eq:pi-part-b-prime2}
%     \sum_{\substack{v\in \ball(u,R)}}\Bigg(\sum_{i\in[k]}\Pr\big\{\delta_{P_i}(u,v) =1 \text{ and } \calE_{vi}\big\}
%     - D_{\beta}\frac{d(u,v)}{R}\,\Pr\{\vee_{P_i}(u,v)=1 \text{ and } \calE_{vi}\}\Bigg)^+\\
%     \leq \sum_{\substack{v\in \ball(u,R) \\ i\in[k]}}\Bigg(\Pr\big\{\delta_{P_i}(u,v) =1 \text{ and } \calE_{vi}\big\} 
%   - D_{\beta}\frac{d(u,v)}{R}\,\Pr\{\vee_{P_i}(u,v)=1 \text{ and } \calE_{vi}\}\Bigg)^+\\
%     \lesssim \beta^q \cdot 
%     \E\Bigg[\sum_{\substack{v\in \ball(u,2R)\\i\in[k]}} \frac{d(u,v)}{R}
%     \cdot \vee_{P_i}(u,v) \cdot \one\{\calE_{vi}\}\Bigg].
% \end{multline}
% \endgroup
% %
\begin{equation} \label{eq:pi-part-b-prime2}
\begin{aligned} 
    &\sum_{\substack{v\in \ball(u,R)}}\Bigg(\sum_{i\in[k]}\Pr\big\{\delta_{P_i}(u,v) =1 \text{ and } \calE_{vi}\big\}
    - D_{\beta}\frac{d(u,v)}{R}\,\Pr\{\vee_{P_i}(u,v)=1 \text{ and } \calE_{vi}\}\Bigg)^+\\
    &\leq \sum_{\substack{v\in \ball(u,R) \\ i\in[k]}}\Bigg(\Pr\big\{\delta_{P_i}(u,v) =1 \text{ and } \calE_{vi}\big\} 
   - D_{\beta}\frac{d(u,v)}{R}\,\Pr\{\vee_{P_i}(u,v)=1 \text{ and } \calE_{vi}\}\Bigg)^+\\
    &\lesssim \beta^q \cdot 
    \E\Bigg[\sum_{\substack{v\in \ball(u,2R)\\i\in[k]}} \frac{d(u,v)}{R}
    \cdot \vee_{P_i}(u,v) \cdot \one\{\calE_{vi}\}\Bigg].
\end{aligned}
\end{equation}
Now consider any $v \in X \setminus \{u\}$. If $u$ and $v$ are separated by the partitioning $\calP$,
then they are separated at some iteration $i$. That is, for some $i\in [k]$:
\begin{itemize}
\item $\calE_{vi}$ happens (in other words, $u$ and $v$ are not clustered at the beginning of iteration $i$)
\item $\delta_{P_i}(u,v) = 1$ (exactly one of them gets clustered in iteration $i$)
\end{itemize}
Further, there is exactly one $i$ such that both events above happen. On the other hand, if $u$ and $v$ are not separated by $\calP$  then $\delta_{P_i}(u,v) = 0$ for all $i\in [k]$. We conclude that
\begin{equation}\label{eq:sum-core-1}
\one\{\calP(u)\neq \calP(v)\big\} = \sum_{i\in [k]}
\one\big\{\delta_{P_i}(u,v) = 1 \text{ and } \calE_{vi}\big\}.
\end{equation}
In particular, the expectations of the expressions on both sides of (\ref{eq:sum-core-1}) are equal: 
\begin{equation}\label{eq:sum-1}
    \Pr\{\calP(u)\neq \calP(v)\} = \sum_{i\in[k]}\Pr\big\{\delta_{P_i}(u,v) =1 \text{ and } \calE_{vi}\big\}.
\end{equation}    
Now consider the first iteration $i$ at which at least one of the vertices $u$ and $v$ gets clustered. Note that (i) event $\calE_{vi}$ happens and (ii) $\vee_{P_i}(u,v) = 1$ (that is, (i) both points $u$ and $v$ are not clustered at the beginning of iteration $i$; (ii) but at least one of them gets clustered in iteration $i$). Further, for $j < i$, 
$\vee_{P_i}(u,v) = 0$ and for $j > i$, $\calE_{vj}$ does not happen. We conclude that event ``$\vee_{P_i}(u,v) = 1$ and $\calE_{vi}$'' happens exactly for one value of $i\in[k]$. Therefore,
\begin{equation}\label{eq:sum-core-2}
\sum_{i\in k} \vee_{P_i}(u,v)\cdot \one\{\calE_{vi}\} = 1
\end{equation}
and
\begin{align}\label{eq:sum-2}
\sum_{i\in [k]}& \Pr\{\vee_{P_i}(u,v) = 1 \text{ and } \calE_{vi}\} = \sum_{i\in [k]} \E[\vee_{P_i}(u,v) \one\{\calE_{vi}\}] = 1.
\end{align}
Plugging (\ref{eq:sum-1}) and (\ref{eq:sum-2}) into (\ref{eq:pi-part-b-prime2}), we obtain
\begin{align*}
    \sum_{\substack{v\in \ball(u,R) }}\Bigg(&\Pr\big\{\calP(u) = \calP(v)\big\} - 
    D_{\beta}\frac{d(u,v)}{R}\Bigg)^+
    \lesssim    \beta^q \cdot 
    \E\bigg[\sum_{\substack{v\in \ball(u,2R)}} \frac{d(u,v)}{R}\bigg].
\end{align*}
We conclude that property (2) holds.
Next, we show that property (3) holds for every $u \in X$. As in the analysis of property (2),
we consider some iteration $i$. Then property (c) of
Theorem~\ref{thm:itr-clustering} guarantees that if $u\in X_i$ then
\begin{align}\label{eq:pi-part-c}
       \sum_{i \in [k]} &\sum_{v \in \ball(u,r)\cap X_i}  \delta_{P_i}(u,v) \lesssim  
        \sum_{i \in [k]} \beta \cdot D^2_{\beta} \cdot \Bigg( \sum_{v\in\ball(u,2R)\cap X_i}\frac{d(u,v)}{R} \cdot \vee_{P_i}(u,v)\Bigg)
\end{align}
We rewrite~(\ref{eq:pi-part-c}) as follows:
\begin{align*}
    \sum_{i \in [k]} \sum_{v \in \ball(u,r)}  \delta_{P_i}(u,v) \cdot \one\{\calE_{vi}\}
    \lesssim\sum_{i \in [k]} \beta \cdot D^2_{\beta} \cdot \Bigg( \sum_{v\in\ball(u,2R)}\frac{d(u,v)}{R} \cdot \vee_{P_i}(u,v)\cdot \one\{\calE_{vi}\}\Bigg).
\end{align*}
Note that this inequality holds for all $u\in X$: if $u\in X_i$, it is equivalent to (\ref{eq:pi-part-c});
if $u\notin X_i$, then both sides are equal to 0, and the inequality trivially holds.
Using formulas (\ref{eq:sum-core-1}) and (\ref{eq:sum-core-2}), we get
$$
       \sum_{v \in \ball(u,r)} \one\big\{\calP(u)\neq \calP(v)\big\}  \lesssim \beta \cdot D^2_{\beta} \cdot  \sum_{v\in\ball(u,2R)}\frac{d(u,v)}{R}.
$$
Therefore, property (3) holds.
\end{proof}
\section{Proof of Theorem~\ref{thm:itr-clustering}}\label{sec:Cluster-select}
\begin{figure}
    \centering
    \scalebox{0.75}{\resizebox {\columnwidth/2} {!} {
\begin{tikzpicture}

\fill[fill=white,draw=black,thick] (0,0) ellipse (5);
% \fill[fill=white,draw=black,thick, dotted] (0,0) ellipse (4.8);
%\fill[fill=black!20!white,draw=black,dotted] (0,0) ellipse (4.2);
\fill[fill=white,draw=black,thick] (0,0) ellipse (4.2);

\fill[fill=black!40!white,draw=black!20!white] (0,0) ellipse (2.7);
\fill[fill=white,draw=white] (0,0) ellipse (2.3);

\draw[draw=black,thick] (0,0) ellipse (2.5);
\fill[fill=white,draw=black,thick] (0,0) ellipse (1);

\node (z) at (0,0) {\huge $z$};

\node (v1) at (1.2,0) {};
 \node (v2) at (-3.1,3.1) {};
\node (v3) at (-3.7,-3.7) {};
\node (v4) at (2.8, -2.8) {};
\node (v5) at (3.7, -3.7) {};
\node (v6) at (6.4, -2.5) {\huge$R_0 = \nicefrac{R}{D_\beta}$};
\node (v9) at (-5.9, -4.4) {\huge$R$};
\node (v10) at (-4.9, 4.9) {\huge$R_1$};
\node (v12) at (1.1, 2.4) {};
\node (v13) at (0.5, 1.1) {};

\node (v11) at (4.9, 4.9) {\huge$t$};

\draw[latex'-latex'] (z) -- node[above] {} (v1);
\draw[latex'-latex'] (z) -- node[above] {} (v2);
\draw[latex'-latex'] (z) -- node[above] {} (v3);
\draw[latex'-latex'] (v4) -- node[above] {} (v5);
\draw[latex'-latex'] (z) -- node[above] {} (v12);

\node (v7) at (3.25, -3.25) {};
\node (v8) at (0.5, -0.1) {};
\node (u1) at (-2.1, -2.2) {};
\node (u2) at (-2.1, 2.2) {};
\draw [black,dotted,thick]   (v6) to[out=200,in=70] (v7);
\draw [black,dotted,thick]   (v6) to[out=160,in=-30] (v8);
\draw [black,dotted,thick]   (v9) to[out=15,in=160] (u1);
\draw [black,dotted,thick]   (v10) to[out=-90,in=180] (u2);
\draw [black,dotted,thick]   (v13)  to[out=90,in=180] (v11);

\node (u3) at (-4.8,0) {};
\node (u4) at (-5,0) {};
%%\draw[latex'-latex'] (u3) -- node[above] {} (u4);
%%\node (u5) at (-5.9, 0) {\huge$r$};
%%\draw [black,dotted,thick]   (u5) to[out=-90,in=-90] (u3);

\node (v14) at (2.4, 1) {};
\node (v15) at (7.2, 1.2) {\huge$\gamma-$light shell};
\node (v16) at (7.8, 0.4) {\huge of width $r$};
\draw [black,dotted,thick]   (v14)  to[out=60,in=180] (v15);

% \node[shape=circle,fill=black] (j) at (-10,0) {};
% \node[text=black] () at (-10,-0.5) {$U$};
% \definecolor{pink}{RGB}{255,179,186}
% \definecolor{blue1}{RGB}{186,225,255}
% \node[draw, fit={(-10.5,2) (-9.5,3)}, inner sep=0pt, label=center:U, fill=pink] (i1) {};
% \node[draw, fit={(-10.5,-2) (-9.5,-3)}, inner sep=0pt, label=center:1, fill=blue1] (i2) {};
% \node[text=black] () at (-10,-4.5) {$1$};

\end{tikzpicture}}

% \begin{tikzpicture}

% \fill[fill=white,draw=black,thick] (0,0) ellipse (5);
% % \fill[fill=white,draw=black,thick] (0,0) ellipse (3);
% \fill[fill=black!20!white,draw=black,dotted] (0,0) ellipse (4.2);
% \fill[fill=white,draw=black,thick] (0,0) ellipse (1);
% \node (z) at (0,0) {\huge $z$};

% \node (v1) at (1.2,0) {};
%  \node (v2) at (-3.1,3.1) {};
% \node (v3) at (-3.7,-3.7) {};
% \node (v4) at (2.8, -2.8) {};
% \node (v5) at (3.7, -3.7) {};
% \node (v6) at (4.9, 4.9) {\huge$R_0$};
% \node (v9) at (4.9, 4.9) {\huge$R$};
% \node (v10) at (4.9, 4.9) {\huge$R_1$};

% \draw[latex'-latex'] (z) -- node[above] {} (v1);
% \draw[latex'-latex'] (z) -- node[above] {\huge $R_1$} (v2);
% \draw[latex'-latex'] (z) -- node[above] {\huge$R$} (v3);
% \draw[latex'-latex'] (v4) -- node[above] {} (v5);
% \draw[latex'-latex'] (v4) -- node[above] {} (v5);

% \node (v7) at (3.25, -3.25) {};
% \node (v8) at (0.5, 0.3) {};
% \draw [black,dotted,thick]   (v6) to[out=-20,in=70] (v7);
% \draw [black,dotted,thick]   (v6) to[out=160,in=70] (v8);

% % \node[shape=circle,fill=black] (j) at (-10,0) {};
}
    \caption{Light Ball}\label{fig:lightball}
    $R > r > 0$, $q\geq 1$, $\beta =r/R$, $D_\beta = 2(q+1) \ln\nicefrac{1}{\beta}$,
    %$r =\nicefrac{\sqrt{\alpha}}{\log (\nicefrac{1}{\alpha})}, R > r$, $\beta =r/R$, $D_\beta = \log \nicefrac{1}{\beta}$,
    $R_0 = R/D_{\beta}$, $R_1 = R - R_0$.
\end{figure}
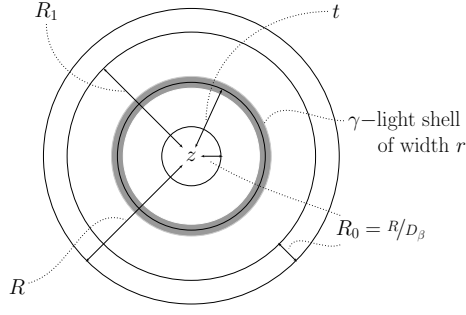
In Section~\ref{sec:iterative_clust}, we describe an iterative approach to finding a probabilistic metric decomposition for Theorem~\ref{thm: decompositon}. In this section, we show how to find one cluster $P$ of the partitioning. Given a metric space $(X,d)$ and positive numbers $r$ and $R$, our algorithm selects a subset $P \subseteq X$ that satisfies the three properties listed in Theorem~~\ref{thm:itr-clustering}. Recall that $\beta = {r}/{R}$, $D_\beta = 2(q+1)\ln \nicefrac{1}{\beta}$, $R_0 = \nicefrac{R}{D_\beta}$, $R_1 = R - R_0$ and $\rho_q(\beta)=(\nicefrac{1}{\beta})^{q+1}$ (see Figure~\ref{fig:lightball}). In this proof, we assume that $\beta = r/R$ is sufficiently small (i.e, $\beta\leq \beta^*_q$ for some small $\beta^*_q=\Theta\big(\frac{1}{(q\ln (q+1)}\big)$ and, consequently, 
$R_0= \nicefrac{R}{D_{\beta}}$ is also small. Specifically, we assume that $r<R_0<R_1<R$ and $R_0+r < R_1/100$.

Our algorithm for selecting the cluster $P$ starts by picking a pivot point $z$ that has the most points within a ball of small radius $R_0$.
That is, $z$ is the optimizer to the following expression:
\begin{align}\label{pvt-condition}
z=\arg \max_{u\in X}|\ball(u,R_0)|.
\end{align}
The algorithm then checks if the ball of a larger radius, $R_1$, around $z$ has significantly more points in it in comparison to the ball of radius $R_0$ around $z$. If the ratio of the number of points in these two balls exceeds $\rho_q(\beta)$, the algorithm selects the set of points $\ball(z, R_1)$ as our cluster $P$. We refer to this case as the ``Heavy Ball'' case. In Section~\ref{subsec:Heavy-ball}, we show that this set $P$ satisfies the properties of Theorem~\ref{thm:itr-clustering}.

If, however, $|\ball(z, R_1)| < \rho_q(\beta) \cdot |\ball(z, R_0)|$, then the algorithm outputs cluster $P = \ball(z, t)$ where $t \in (0, R]$ is chosen as follows. First, the algorithm finds the set $S$ of all radii $s\in(3R_0,R_1]$ for which the set $P = \ball(z,s)$ satisfies Definition~\ref{def:set_S}. Then, it chooses a random radius $t$ in $S$ (non-uniformly) so that random set $P = \ball(z,t)$ satisfies property~(b) of Theorem~\ref{thm:itr-clustering}. In Section~\ref{sec:L_inf}, we discuss how to find the set $S$ and show that $\mu(S)\geq R/2$ (where $\mu(S)$ is the Lebesgue measure of set $S$). Finally, in Section~\ref{sec:L_one}, we describe a procedure for choosing a random radius $t$ in $S$.
\subsection{Useful Observations}
In this section, we prove several lemmas which we will use  for analyzing both the ``Heavy Ball'' and ``Light Ball'' cases.
First, we show an inequality that will help us lower bound the right hand sides in inequalities (b) and (c) of Theorem~\ref{thm:itr-clustering}.

\begin{lemma}\label{lem:item-Y-helper}
Assume that $z$ is chosen according to~(\ref{pvt-condition}). Consider $t$ in $(3R_0,R_1]$ and $u$ in $X$ with $d(z,u)\in [2R_0,R]$. Let $P=\ball(z,t)$. Denote:
$$Y_P = 
\sum_{v\in \ball(u,2R)} \frac{d(u,v)\vee_P(u,v)}{R}.$$
Then, 
$|P|\leq 2 D_{\beta} Y_P.$
\end{lemma}
\noindent\textbf{Remark: } Note that in Theorem~\ref{thm:itr-clustering}, the right side of inequality (b)  equals $\beta^q \E[Y_P]$, and the right side of inequality (c) equals $\beta \cdot D_\beta^2 \cdot Y_p$.

\begin{proof}
Observe that $P\subset \ball(u,2R)$. Hence,
$$
Y_P = \sum_{v\in \ball(u,2R)} \frac{d(u,v)\vee_P(u,v)}{R} \geq
\sum_{v\in P} \frac{d(u,v)}{R}.
$$
For every $v\in P\setminus \ball(u,R_0)$, we have 
$\frac{d(u,v)}{R}\geq \frac{R_0}{R}=D_{\beta}^{-1}$.
Thus, 
$$Y_P\geq D_{\beta}^{-1}\; |P\setminus \ball(u,R_0)|.$$
We need to lower bound the size of $P\setminus\ball(u,R_0)$. On the one hand, we have
\begin{align*}
|P\setminus \ball(u,R_0)|&\geq 
|P|-|\ball(u,R_0)|\\
&\geq |P|-|\ball(z,R_0)|.
\end{align*}
Here, we used that $\ball(z,R_0)$ is the largest ball of radius $R_0$ in $X$.  On the other hand, $\ball(z, R_0)\subset P\setminus \ball(u,R_0)$, since $d(z,u)\geq 2R_0$. Thus,
$|P\setminus \ball(u,R_0)|\geq |\ball(z, R_0)|$.
Combining two bounds on $|P\setminus \ball(u,R_0)|$,
we get the desired inequality 
$|P\setminus \ball(u,R_0)|\geq |P|/2$.
\end{proof}

We now provide a lemma that will help us verify property (b) of Theorem~\ref{thm:itr-clustering} for that point $u$.

\begin{lemma}\label{lem:item-b-helper}
Consider an arbitrary probability distribution of $t$ in $(3R_0,R_1]$. Let $P=\ball(z,t)$, where $z$ is chosen according to~(\ref{pvt-condition}). 
If for each point $u\in \ball(z, R)$ at least one of the following two conditions holds, then $P$ satisfies property (b) of Theorem~\ref{thm:itr-clustering} for all points $u$ in $X$.\\
\noindent Condition I:
\begin{equation}\label{eq:helper-b1}
\Pr\{t\geq d(z,u)-R_0\}
\lesssim 
\beta^{q+1}\cdot \frac{\E|\ball(z,t)| }{|\ball(z,R_0)|}.
\end{equation}\label{eq:cond-I}
\noindent \noindent Condition II: For every $v\in \ball(u,R_0)$, we have
\begin{align}\label{eq:helper-b2}
\Pr\big\{&\delta_P(u,v) = 1\big\} -
D_{\beta}\frac{d(u,v)}{R} \Pr\{\vee_P(u,v)=1\}
\lesssim \beta^{q+1}
.
\end{align}
\end{lemma}

\noindent\textbf{Remark:} 
This lemma makes the argument about the distribution of $t$ from the proof overview section (Section~\ref{sec:proof_overview}) more precise. As we discuss in subsection Light Ball~\ref{sec:light_ball_main}, we have chosen the distribution of $t$ (Cumulative distribution function $F$, Definition~\ref{def:F}) to satisfy two properties: (Property I) the probability that $u$ and $v$ are separated by $P$ is upper bounded by the probability that $u$ or $v$ is in $P$ times $O(D_{\beta})$; and (Property II) The probability that $t$ is close to $\pi_S^{inv}(R/2)$ is small. Thus, Condition I of Lemma~\ref{lem:item-b-helper} holds for $u$ with $d(z,u)$ that are sufficiently close to $\pi_S^{inv}(R/2)$, and Condition II holds for $u$ with for smaller values of $d(z,u)$.
\begin{proof}
Consider one term from the left hand side of property~(b) of Theorem~\ref{thm:itr-clustering} for some $u$ in $X$:
$$
\Big(\Pr\big\{\delta_P(u,v) =1 \big\} 
- D_{\beta}\frac{d(u,v)}{R} \cdot \Pr\{\vee_P(u,v)=1\}\Big)^+.
$$
Note that $\{\delta_P(u,v) = 1\}$ denotes the event that \textit{exactly one of the points} $u$ \textit{and} $v$ lies in $P$;
whereas $\{\vee_P(u,v) = 1\}$ denotes the event that \textit{at least} one of $u$ and $v$
lies in $P$. Thus, $\Pr\{\vee_P(u,v) = 1\} \geq \Pr\{\delta_P(u,v) = 1\}$. Hence, this expression is positive only if $D_{\beta}\cdot d(u,v)/R < 1$, which is equivalent to
$$d(u,v) < R/D_{\beta} = R_0.$$
Thus, in the left hand side of property~(b), we can consider only $v$ in $\ball(u,R_0)$ 
(rather than $\ball(u,R)$). Moreover, if $d(z,u)> R$, then for all $v \in \ball(u,R_0)$, we have $d(z,v) > R-R_0=R_1$ and, consequently, $\delta_P(u,v)=0$. Therefore, for such $u$, the left hand side of property (b) equals $0$, and the inequality (b) holds trivially. We will thus assume that $d(z,u)\leq R$ (which is equivalent to $u\in \ball(z,R)$). Similarly, since $t > 3R_0$, we will assume that $d(z,u)\geq 2R_0$ (otherwise, $u\in P$ and every $v\in \ball(u,R_0)$ is in $P$, and thus $\delta_P(u,v) = 0$).

We now show that if Condition I or II of~Lemma~\ref{lem:item-b-helper} holds for $u\in \ball(z,R)$ then property~(b) is satisfied for that $u$.

I. Suppose, the first condition holds for $u\in \ball(z,R)$. If $\delta_P(u,v) = 1$ then either $u\in P,\; v\not\in P$ or $v\in P,\; u\not\in P$. In the former case, $t\geq d(z,u)$; in the latter case, $t\geq d(z,v) \geq d(z, u) - R_0$. In either case, $t\geq d(z, u) - R_0$.
Using that $|\ball(u, R_0)| \leq |\ball(z, R_0)|$ by our choice of $z$ (see (\ref{pvt-condition})), we bound the left hand side of (b) as follows
\begin{align*}
\sum_{v\in \ball(u, R_0)}&\Pr\{\delta_P(u,v)=1\}
\\ &\leq 
|\ball(u, R_0)|\Pr\{t \geq d(z,u)-R_0\}
\\ &\leq 
|\ball(z, R_0)|\Pr\{t \geq d(z,u)-R_0\}
\end{align*}
%Here, we used that $d(z,v)\geq d(z,u)-R_0$. Thus, if 
%$\delta_P(u,v)=1$, then we must have $t \geq d(z,u)-R_0$.
We now use the inequality from Condition I to get the bound
\begin{align*}
\sum_{v\in \ball(u, R_0)}\Pr\{\delta_P(u,v)=1\}&
\lesssim \beta^{q+1} \E|\ball(z,t)|\\
&=\beta^{q+1}\E|P|.    
\end{align*}
Finally, by Lemma~\ref{lem:item-Y-helper}, we have the following bound on $\beta^{q+1}\E|P|$:
\begin{equation}\label{eq:beta-Y-P}
\beta^{q+1} \E|P|\leq 
\beta^{q+1}\cdot 2D_{\beta} \E[Y_P]
\leq
\beta^q\E[Y_P].
\end{equation}
Here, we used that $2\beta D_{\beta} = 2r/R_0 <1$ by our choice of $\beta^*_{q}$. The right hand side of the inequality in property (b) equals $\beta^q\E[Y_P]$. Thus, property (b) holds. 

\noindent II. Suppose now that the second condition holds for $u\in \ball(z,R)$. Then, each term in the left hand side of (b) is upper bounded by $O(\beta^{q+1})$. Hence, the entire sum is upper bounded by $O(\beta^{q+1}|\ball(u,R_0)|)$, which in turn is upper bounded by $O(\beta^{q+1}|\ball(z,R_0)|)$. Then,
\begin{align*}
\beta^{q+1}|\ball(z, R_0)| &\leq \beta^{q+1}|\ball(z, t)|
= \beta^{q+1} \E|P| \\
&\leq \beta^q\E[Y_P].
\end{align*}
The last inequality follows from~(\ref{eq:beta-Y-P}). We conclude that property~(b) of Theorem~\ref{thm:itr-clustering} holds.
\end{proof}

\subsection{Heavy Ball Case}\label{subsec:Heavy-ball}

In this subsection, we analyze the case when $|\ball(z, R_1)| \geq \rho_q(\beta) \cdot |\ball(z, R_0)|$. If this condition is met, then the algorithm outputs $P = \ball(z, R_1)$. We will show that Theorem~\ref{thm:itr-clustering} holds for such a cluster $P$. 
%\jnote{Shouldn't this be $P = \ball(z, R)$?}

We first prove properties (a) and (b). Since the radius of $P$ is $R_1\leq R$, its diameter is at most $2R$. So property (a) of Theorem~\ref{thm:itr-clustering} holds. To show property (b), we apply Lemma~\ref{lem:item-b-helper} (item I) with $t=R_1$. Trivially, $\Pr\{t\geq d(z,u)-R_0\}\leq 1$ and $\E|\ball(z,t)| = |\ball(z,R_1)|\geq \rho_q(\beta) |\ball(z,R_0)|$. Thus, (\ref{eq:cond-I}) is satisfied and property (b) also holds.

We now show property (c) of Theorem~\ref{thm:itr-clustering}. Observe that
if $d(z,u)\notin [R_1-r,R_1+r]$, then $\delta_P(u,v)=0$ for all $v\in \ball(u,r)$ (because $P=\ball(z,R_1)$). Hence, for such $u$, property (c) holds. Thus, we assume that $u\in [R_1-r,R_1+r]\subseteq [2R_0,R]$.

We bound the left hand side of (c) as follows:
\begin{align*}\sum_{v\in \ball(u,r)}\delta(u,v)&\leq 
|\ball(u,r)|\leq |\ball(u,R_0)| \\ &\qquad \qquad \qquad \leq |\ball(z,R_0)|,
\end{align*}
here we first use that $r \leq R_0$ and then that $z$ satisfies (\ref{pvt-condition}).
Since we are in the Heavy Ball Case, we have $|P|\geq \rho_q(\beta)|\ball(z,R_0)|$. Therefore,
$$\sum_{v\in \ball(u,r)}\delta(u,v)\leq |P|/\rho_q(\beta).$$
By Lemma~\ref{lem:item-Y-helper}, the right hand side is upper bounded by 
$$2D_{\beta} Y_P/\rho_{q}(\beta) = 
2D_{\beta} \, \beta^{q+1} \, Y_P
\lesssim  \beta D^2_{\beta}\, Y_P.$$
The right hand side of (c) equals 
$\beta D^2_{\beta}\, Y_P$. Hence, property (c) is satisfied.

Thus we have shown that Theorem~\ref{thm:itr-clustering} holds for the case of Heavy Balls. To complete the proof, we show that Theorem~\ref{thm:itr-clustering} also holds for the case 
of Light Balls -- we give this proof in Section~\ref{sec:light-ball}. 

\subsection{Light Ball Case}\label{sec:light-ball}

%\begin{figure}
%    \centering
%    \scalebox{2}{\input{figures/image-lightball}}
%    \caption{Light Ball}\label{fig:lightball}
%    $R > r > 0$, $q\geq 1$, $\beta =r/R$, $D_\beta = 2(q+1) \ln\nicefrac{1}{\beta}$,
    %$r =\nicefrac{\sqrt{\alpha}}{\log (\nicefrac{1}{\alpha})}, R > r$, $\beta =r/R$, $D_\beta = \log \nicefrac{1}{\beta}$,
%    $R_0 = R/D_{\beta}$, $R_1 = R - R_0$.
%\end{figure}
% \knote{slightly updated the label in Figure~\ref{fig:lightball}}
We now consider the case when 
$|\ball(z,R_1)|\leq \rho_q(\beta)\cdot |\ball(z,R_0)|$. 
Recall that $S$ is the set of all radii $s\in (3R_0,R_1]$ for which property (c)  of Theorem~\ref{thm:itr-clustering} holds (Definition~\ref{def:set_S}). The set $S$ can be found in polynomial time since the number of distinct balls $\ball(z,s)$ is upper bounded by the number of points in the metric space. We now recall map $\pi_S$ used in Algorithm~\ref{alg:ball-select}. 

\noindent{\textbf{Map $\pi_S$}}.
%%\subsubsection{Map $\pi_S$}\label{sec:map-pi}
In Section~\ref{sec:single_clust}, we  define a measure preserving transformation $\pi_S$ that maps a given measurable set  $S\subset [0,R]$ to the interval $[0,\mu(S)]$ (Definition~\ref{def:function-pi}). We  need this transformation in Algorithm~\ref{alg:ball-select}. If $S$ is the union of several disjoint intervals (as in our algorithm) then $\pi_S$ simply pushes all intervals to the left so that every two consecutive intervals touch each other.
We show the following lemma.
\begin{lemma}\label{lem:function-pi}
For any measurable set $S$, $\pi_S$ is a continuous non-decreasing 1-Lipschitz function, and $\pi^{inv}_S$ is a strictly 
increasing function defined for all $y$ in $[0,\mu(S)]$.
Moreover, there exists a set $Z_0$ of measure zero such that for all $y \in [0,\mu(S)]\setminus Z_0$, we have $\pi_S^{inv} (y)\in S$.
\end{lemma}
\begin{proof}
Note that $\pi_S^{inv}(y)$ is a right inverse for $\pi_S(x)$: $\pi_S(\pi_S^{inv}(y)) = y$ (but not necessarily a left inverse).
Let 
$$I_S(x) = 
\begin{cases}
    1,&\text{if } x\in S\\
    0,&\text{otherwise}
\end{cases}
$$
be the indicator function of set $S$.
Then $\pi_S(x) = \int_0^x I_S(t) dt$ (we use Lebesgue integration here). Since $0\leq I_S(t) \leq 1$, function $\pi_S$ is non-decreasing, $1$-Lipschitz, and absolutely continuous. By the Lebesgue differentiation theorem, $\pi_S(x)$ is almost everywhere differentiable and $\frac{d\pi_S(x)}{dx} = I_S(x)$ almost everywhere. Let $X_0 = [0,R]\setminus S$ and $Z_0 = \pi_S(X_0)$. Since $\pi_S$ is absolutely continuous and $I_S(x) = 0$ for $x\in X_0$, we have
$$\mu(Z_0)\leq\int_{X_0}\frac{d\pi_S(x)}{dx}dx = \int_{X_0}I_S(x) dx = 0.$$
Now if $y\notin Z_0$, then $\pi_S(\pi_S^{inv}(y)) = y\notin Z_0$, thus $\pi_S^{inv}(y)\notin X_0$ or, equivalently,
$\pi_S^{inv}(y)\in S$, as required.

Finally, we verify that $\pi_S^{inv}$ is strictly increasing. Consider $a,b \in [0,\mu(S)]$ with $a < b$. Note that $a = \pi_S(\pi_S^{inv}(a))$ and $b = \pi_S(\pi_S^{inv}(b))$. Thus, $\pi_S(\pi_S^{inv}(a))< \pi_S(\pi_S^{inv}(b))$. Since $\pi_S$ is non-decreasing, 
$\pi_S^{inv}(a) < \pi_S^{inv}(b)$.
\end{proof}

Note that if $S$ is a union of finitely many disjoint open intervals, then $Z_0$ is the image of the endpoints of those intervals under $\pi_S$.

%\noindent\textbf{Clusters satisfying property (c).}
\subsection{Clusters Satisfying Property (c) of Theorem~\ref{thm:itr-clustering}}
\label{sec:L_inf}
We first show that if $|\ball(z,R_1)| < \rho_q(\beta) \cdot |\ball(z,R_0)|$, then 
$\mu(S)\geq R/2$. To this end, we define a ball with a $\gamma$-light shell of width $r$.

\begin{definition}
We say that the ball of radius $t\geq r$ around $z$ has a $\gamma$-light shell of width $r$ if
$$|\ball(z,t+r)| - |\ball(z,t-r)| \leq \gamma
 \int_0^{t-r} |\ball(z,x)|\, dx.$$
\end{definition}

We let $S_{\gamma}$ be the set of all radii $t$ in the range $(3R_0, R_1]$ 
such that $\ball(z, t)$ has a $\gamma$-light shell of width $r$. We now show that (a) $S_{\gamma}\subset S$ and (b) $\mu(S_{\gamma})\geq {R}/{2}$ for 
$\gamma = 25r/R_0^2$.
%%\begin{equation}\label{eq:def-gamma}
%%\gamma = 4r/R_0^2,
%%\end{equation}
and, therefore, $\mu(S)\geq {R}/{2}$.

\begin{lemma}\label{lem:S-gamma-subset-S}
We have $S_{\gamma}\subset S$.
\end{lemma}
%\section{Proof of Lemma~\ref{lem:S-gamma-subset-S}}\label{sec:lem:S-gamma-subset-S}
\begin{proof}
Consider a number $t$ from $S_{\gamma}$ and the ball of radius $t$ around $z$: $P=\ball(z,t)$. 
Let us pick an arbitrary point $u$. We are going to prove that inequality~(\ref{eq:newC}) holds and therefore $t\in S$. Consider $v\in \ball(u,r)$.
Observe that $\delta_P(u,v) = 1$ only if both $u$ and $v$ belong to the $r$ neighborhoods of $P$ and $X\setminus P$. Thus, if $\delta_P(u,v) = 1$, we must have
$d(z,u),d(z,v)\in [t-r,t+r]$. 
If $d(z,u)\notin [t-r,t+r]$, then 
the left side of~(\ref{eq:newC})
equals $0$, and we are done. Hence, we can assume that 
$d(z,u)\in [t-r,t+r]$. 

Using the observation above, we bound
the left hand side of (\ref{eq:newC}) as
$$\sum_{v\in \ball(u,r)}\delta_P(u,v)\leq 
|\ball(z,t+r)|-|\ball(z,t-r)|.$$

We now need to lower bound the right hand side of~(\ref{eq:newC}). Note 
that $\ball(u,2R)$ contains $\ball(z,t)$, since 
$$d(z,u)\leq t+r\leq R_1 + r = R-R_0+r<R,$$
and $t < R$. Thus, 
$$\sum_{v\in\ball(u,2R)} \frac{d(u,v)}{R}\vee_p(u,v)
\geq \frac{1}{R} \sum_{v\in\ball(z,t)} d(u,v)\vee_p(u,v).$$
For all $v\in \ball(z,t)\equiv P$, we have $\vee_P(u,v)=1$. Hence,
\begin{align}\label{eq: total_distance_eq}
    \sum_{v\in\ball(z,t)} d(u,v)\vee_p(u,v) = \sum_{v\in\ball(z,t)} d(u,v)
\end{align}
By the triangle inequality, we have
$$d(u,v)\geq (d(z,u)-d(z,v))^+\geq ((t-r) -d(z,v))^+.$$
Observe that
$$((t-r) -d(z,v))^+= \int_0^{t-r} \one\big\{d(z,v) \leq x\big\}\; dx.$$
Hence,~(\ref{eq: total_distance_eq}) is lower bounded by
\begin{multline*}
\sum_{v\in\ball(z,t)}\int_0^{t-r} \one\big\{d(z,v) \leq x\big\}\; dx
=
\int_0^{t-r} \sum_{v\in\ball(z,t)} \one\big\{d(z,v) \leq x\big\}\; dx
=
\int_0^{t-r} |\ball(z,x)|\; dx.
\end{multline*}
Since the ball of radius $t$ has a $\gamma$-light shell of width $r$, the expression above is, in turn,  lower bounded by 
$$\frac{|\ball(z,t+r)|-|\ball(z,t-r)|}{\gamma}.$$
Thus, the right hand side of inequality (\ref{eq:newC}) is lower bounded by 
$$
 \frac{25\beta D_{\beta}^2}{R}\cdot 
\frac{|\ball(z,t+r)|-|\ball(z,t-r)|}{\gamma}.
$$
This completes the proof of Lemma~\ref{lem:S-gamma-subset-S}, since
$$
\frac{25\beta D_{\beta}^2}{R}\cdot \frac{1}{\gamma}  = 
\frac{25\beta D_{\beta}^2  R_0^2}{25R\cdot r}= \frac{\left(\nicefrac{r}{R}\right) D_\beta^2 \left(\nicefrac{R}{D_\beta}\right)^2}{Rr}= 1.
$$
\end{proof}

%%We prove this lemma in Appendix~\ref{sec:lem:S-gamma-subset-S}.

\begin{lemma}\label{lem:mu-S-large}
We have $\mu\big(S_{\gamma}\big)\geq \nicefrac{R}{2}$.
\end{lemma}
To prove this lemma, we use the following result from Appendix~\ref{sec:proof-lem-Phi-S}. 
\begin{lemma}\label{lem:Phi-S}
Consider a non-decreasing function $\Phi:[0,R]\to\bbR$ with $\Phi(0) = 1$ and $R>0$. Let $r \in (0,R]$ and $\gamma \leq (0,1/r]$.
Then, for the subset $S$ of numbers $t\in [0,R-r]$ 
for which inequality 
\begin{equation}\label{ineq:Phi-S}
\Phi(t+r)\geq \Phi(t) + \gamma \int_0^{t} \Phi(x) dx
\end{equation}
holds, we have $\Phi(R)\geq e^{\eta \mu(S)-1}$, where 
$\eta = \sqrt{\nicefrac{\gamma}{(e-1)r}}$, and 
$\mu(S)$ is the measure of set $S$.
\end{lemma}

\begin{proof}[Proof of Lemma~\ref{lem:mu-S-large}]
We apply Lemma~\ref{lem:Phi-S} to the function
$$\Phi(t) = \frac{|\ball(z,t+3R_0)|}{|\ball(z,3R_0)|}$$
with parameters $r'= 2r$, $R'=R_1-3R_0-r$, and $\gamma=25r/R^2_0$. Note that to be able to apply Lemma~\ref{lem:Phi-S} we need $\gamma<1/r'$ which is equivalent to $\beta D_{\beta}<\nicefrac{1}{5\sqrt{2}}$. The latter holds due to $\beta$ being sufficiently small, i.e., $\beta\leq \Theta\big(\frac{1}{(q\ln (q+1)}\big)$. Observe that $\Phi(0)=1$ and
\begin{align*}
\Phi(R') & \leq \frac{|\ball(z,R_1)|}{|\ball(z,3R_0)|} 
\leq \frac{\rho_q(\beta)|\ball(z,R_0)|}{|\ball(z,R_0)|} 
= \rho_q(\beta).
\end{align*}
Here, we used that the $\ball(z,R_1)$ is light. From Lemma~\ref{lem:Phi-S}, we get that 
$\Phi(R')\geq e^{\eta' \mu(S')-1}$, where 
$\eta' = \sqrt{\nicefrac{\gamma}{(e-1)r'}}$, and $S'$ is the 
set of $t$ for which Inequality~(\ref{ineq:Phi-S}) holds. Thus,
\begin{align*}
\mu(S')&\leq \frac{1+ \ln \Phi(R')}{\eta'} \leq \frac{1 + \ln \rho_q(\beta)}{\eta'} = \frac{1+D_{\beta}/2}{\eta'}\\
&= 
\sqrt{\frac{(e-1)r'}{\gamma}} \cdot (1+D_{\beta}/2)\\
&= \sqrt{\frac{2(e-1)r}{25r}}\cdot
R_0 \cdot (1+D_{\beta}/2)\\
&= \sqrt{\frac{2(e-1)}{25}}\cdot 
(R_0+R/2)< 0.4 (R + R_0).
\end{align*}
where we used $R_0 \cdot D_\beta = R$ and that $\sqrt{2(e-1)} < 2$. Therefore for the measure of the set $S''=[0,R']\setminus S'$ is at least
$\mu(S'')\geq ((R-R_0)-3R_0-r) - 0.4 (R + R_0) \geq R/2$.
Here, we relied on our assumption that 
$R_0 + r < R_1/100$.

We claim that $S'' + 3R_0 + r\subset S_{\gamma}$. Consider an arbitrary $t\in S''$. First, observe that 
$t+3R_0+r\in (3R_0,R_1]$. Then,
\begin{align*}
\frac{|\ball(z,t+3R_0 + r')|}{|\ball(z,3R_0)|} - \frac{|\ball(z,t+3R_0)|}{|\ball(z,3R_0)|} &=\Phi(t+r') - \Phi(t) \\ 
&< \gamma\int_0^t\Phi(x) dx = 
\gamma \int_0^{t} \frac{|\ball(z,x+3R_0)|}{|\ball(z,3R_0)|}  dx.
\end{align*}
For $t' = t+ 3R_0 + r$, we get
\begin{align*}
|\ball(z,t'+r)|  - |\ball(z,t'-r)| 
&< \gamma \int\displaylimits_{\mathclap{0}}^{t'-3R_0-r} |\ball(z,x+3R_0)|  dx\\
&= \gamma \int\displaylimits_{3R_0}^{t'-r} |\ball(z,x)|  dx
< \gamma \int\displaylimits_{0}^{t'-r} |\ball(z,x)|  dx.
\end{align*}
Thus, $t'\in S_{\gamma}$. This finishes the proof.
\end{proof}
Lemma~\ref{lem:mu-S-large} together with Lemma~\ref{lem:S-gamma-subset-S} imply the following corollary.
\begin{corollary}
Let $S$ be the set defined in Definition~\ref{def:set_S}. Then, $\mu(S)\geq R/2.$
\end{corollary}

\subsection{Clusters Satisfying Property (b) of Theorem~\ref{thm:itr-clustering}}\label{sec:L_one}

We now show how to choose a random $t\in S$, so that the random cluster $P = \ball(z,t)$ 
satisfies property (b) of Theorem~\ref{thm:itr-clustering}. We first choose 
a random $x\in[0,R/2]$ with the cumulative distribution function $F(x)$ defined in Definition~\ref{def:F},
%$$
%F(x) = 
%\frac{1-e^{-\nicefrac{x}{R_0}}}{1-e^{-\nicefrac{R_1}{2R_0}}},
%$$
and then let $t=\pi^{inv}_S(x)$, where 
$S\subset (3R_0, R_1]$ is the set obtained in the previous section. Note that by Lemma~\ref{lem:function-pi}, 
$t=\pi^{inv}_S(x)\in S$ with probability $1$, since $\Pr\{x\in Z_0\}=0$ (see Lemma~\ref{lem:function-pi}).

To show that property (b) is satisfied, we verify that for every $u$ in $\ball(z, R)$, Condition I or Condition II of Lemma~\ref{lem:item-b-helper} holds.

Pick a point $u$ in $\ball(z, R)$. We consider two cases: $\pi_S(d(z,u))> R/2-R_0$ and $\pi_S(d(z,u))\leq R/2-R_0$. We prove that $u$ satisfies Condition I of Lemma~\ref{lem:item-b-helper} in the former case and Condition II in the latter case.

\noindent \textbf{First case: $\pi_S(d(z,u)) > R/2-R_0$}. Write,
$$\Pr\{t\geq d(z,u)-R_0\} = 
\Pr\{x\geq \pi_S(d(z,u)-R_0)\}.$$
Since $\pi_S$ is a 1-Lipschitz function, we have
$$\pi_S(d(z,u)-R_0)\geq \pi_S(d(z,u))-R_0\geq R/2-2R_0.$$
Therefore,
$$\Pr\{t\geq d(z,u)-R_0\} \leq 1 - F(R/2-2R_0).$$
%By Claim~\ref{cl:apx:tailF} (see Appendix~\ref{sec:TailF}):
%$$1 - F(R/2-2R_0)\lesssim \beta^{q+1}.$$
We prove the following claim.
\begin{claim}\label{cl:apx:tailF} We have
$$1 - F(R/2-2R_0)\lesssim \beta^{q+1}.$$
\end{claim}
\begin{proof} Write:
$$F(R/2-2R_0)
=\frac{1-e^\frac{-R}{2R_0}\cdot e^\frac{2R_0}{R_0}}{1-e^\frac{-R}{2R_0}}
=\frac{1-e^2 e^{\nicefrac{-D_{\beta}}{2}}}{1-
e^{\nicefrac{-D_{\beta}}{2}}}.
$$
Note that $e^{\nicefrac{-D_{\beta}}{2}}=
\beta^{q+1}$.
Then,
$$1 - F(R/2-2R_0) = \frac{(e^2 -1)}{1-\beta^{q+1}}
\cdot \beta^{q+1}.$$
Since the denominator of the right hand side is greater than $\nicefrac{1}{2}$ (recall that we assume that $\beta$ is sufficiently small),
we have
$1 - F(R/2-2R_0) \lesssim \beta^{q+1}$. \end{proof}
Claim~\ref{cl:apx:tailF} finishes the analysis of the first case, since $|\ball(z,t)|/|\ball(z,R_0)|\geq 1$ for every value of $t\geq R_0$.

\noindent \textbf{Second case: $\pi_S(d(z,u))\leq R/2-R_0$}. In this case, for every $v\in \ball(u,R_0)$, we have
$$\pi_S(d(z,v))\leq \pi_S(d(z,u)+R_0) \leq R/2.$$
Here, we used that $\pi_S$ is a 1-Lipschitz function. We claim that inequality~(\ref{eq:helper-b2}) holds for every two points 
$v_1,v_2\in X$ with $\pi_S(d(z,v_1)),\pi_S(d(z,v_2))\leq R/2$ and  $d(v_1,v_2)\leq R_0$. In particular, it holds for $v_1=u$ and $v_2=v$. Without loss of generality assume, that $d(z,v_1)\leq d(z,v_2)$.
Then,
%%\jnote{We should also add the condition $d(v_1,v_2)\leq R_0$ since we use it at the end of the proof. }
\begin{align*}
\Pr\{\delta_P(v_1,v_2)=1\} &= \Pr\{d(z,v_1)\leq t < d(z,v_2)\}\\
&=\Pr\{\pi_S(d(z,v_1))\leq x < \pi_S(d(z,v_2))\}\\
&= F(\pi_S(d(z,v_2))) - F(\pi_S(d(z,v_1))).
\end{align*}
Here, we used that random variable $x$ has distribution function $F$. 
%%We now show the following claim.
We show the following claim.
\begin{claim}\label{cl:ineq-on-F}
For all $x_1 \leq x_2$ in the range $[0,R/2]$, we have 
$$F(x_2)-F(x_1) \leq D_{\beta}\cdot \frac{(x_2-x_1)}{R}\cdot \big(1 - F(x_1) + 2\beta^{q+1}\big). $$
\end{claim}
\begin{proof}
We have
\begin{align*}
F(x_2)-F(x_1) &= \int_{x_1}^{x_2} F'(x) dx\\
&\leq (x_2 - x_1) \max_{x\in [x_1,x_2]}F'(x)\\
&= (x_2 - x_1) \cdot \frac{e^{-\nicefrac{x_1}{R_0}}/R_0}{1-e^{-\nicefrac{R}{2R_0}}}\\
&= D_{\beta}\cdot \frac{(x_2 - x_1)}{R} \cdot \frac{e^{-\nicefrac{x_1}{R_0}}}{1-e^{-\nicefrac{R}{2R_0}}}.
\end{align*}
Here, we used that $R_0 = R/D_{\beta}$. We now need to upper bound the third term on the right hand side:
\begin{align*}
\frac{e^{-\nicefrac{x_1}{R_0}}}{1-e^{-\nicefrac{R}{2R_0}}}
&= 1 - \frac{(1 -
e^{-\nicefrac{R}{2R_0}}) -
e^{-\nicefrac{x_1}{R_0}
}}{1-e^{-\nicefrac{R}{2R_0}}}\\
&= 1 - F(x_1) +
\frac{e^{-\nicefrac{R}{2R_0}}}{1-e^{-\nicefrac{R}{2R_0}}}.
\end{align*}
As in Claim~\ref{cl:apx:tailF}, let us use that $e^{-\nicefrac{R}{2R_0}} = \beta^{q+1}$ and $1-\beta^{q+1}\geq 1/2$ to get
$$\frac{e^{-\nicefrac{x_1}{R_0}}}{1-e^{-\nicefrac{R}{2R_0}}}\leq 1 - F(x_1)+2\beta^{q+1}.
$$
Combining the bounds above, we get the following inequality:
$$
F(x_2)-F(x_1)\leq
D_{\beta}\cdot\frac{x_2 - x_1}{R} \cdot    
\Big(1 - F(x_1) + 2\beta^{q+1} \Big).
$$

\end{proof}
Using Claim~\ref{cl:ineq-on-F} and the inequality
\begin{align*}
\pi_S(d(z,v_2))) - \pi_S(d(z,v_1)&\leq d(z,v_2) - d(z,v_1)\\
&\leq d(v_1,v_2),
\end{align*}
we derive the following upper bound
\begin{multline*}
\Pr\{\delta_P(v_1,v_2)=1\} \leq \\ \leq D_{\beta}
\frac{d(v_1,v_2)}{R}\cdot \big(1 - F(\pi_S(d(z,v_1)))) + 2\beta^{q+1}\big).
\end{multline*}
Then,
\begin{align*}
\Pr\{\vee_P(v_1,v_2) = 1\} &=  \Pr\{d(z,v_1)\leq t\}\\
&= 1- F(\pi_S(d(z,v_1))).
\end{align*}
Therefore,
\begin{align*}
\Pr\{\delta_P(v_1,v_2)=1\} 
- D_{\beta}\frac{d(v_1,v_2)}{R}\cdot \Pr\{\vee_P(v_1,v_2) = 1\}
\leq 2D_{\beta}
\frac{d(v_1,v_2)}{R}\beta^{q+1}.
\end{align*}

Thus, the left hand side of (\ref{eq:helper-b2}) is upper bounded by 
$$2D_{\beta}\cdot\beta^{q+1} \cdot \frac{d(v_1,v_2)}{R}\leq 
2\beta^{q+1}.$$
Here, we use that $d(v_1,v_2)\leq R_0$ and $R_0=R/D_{\beta}$.
%Here, we use that $v\in \ball(u,R_0)$ and, thus $d(u,v)\leq R_0=R/D_{\beta}$.

\bibliographystyle{icml2021}
\bibliography{references}

\appendix

\clearpage

\section{Proof of Lemma~\ref{lem:Phi-S}}\label{sec:proof-lem-Phi-S}
We first prove Lemma~\ref{lem:Phi-S} for the case when 
$S$ is a measure zero set. Specifically, we show the following lemma.
\begin{lemma}\label{lem:Phi}
Suppose, a non-decreasing function $\Phi:[0,R]\to\bbR$ with $\Phi(0) = 1$ satisfies the following inequality for all $t\in [0,R-r]\setminus Y_0$,
where set $Y_0$ has measure zero:
\begin{equation}\label{ineq:Phi}
\Phi(t+r)\geq \Phi(t) + \gamma \int_0^{t} \Phi(x) dx,
\end{equation}
for some $R>0$, $r \in (0,R/2]$ and $\gamma \in (0,1/r]$, then 
$\Phi(t) \geq \max \{e^{\eta t - 1},1\}$ for all $t\in [0,R]$, where
$\eta = \sqrt{\nicefrac{\gamma}{(e-1)r}}$. Consequently, we have $\Phi(R)\geq e^{\eta R-1}$.
\end{lemma}
\begin{proof} Since
$\Phi(0) = 1$ and $\Phi(t)$ is non-decreasing, we have $\Phi(t)\geq 1$ for all $t\geq 0$. We now prove that
$\Phi(t)\geq e^{\eta t-1}$. We establish this inequality by induction. The inductive hypothesis is that this inequality holds
for $t\in [0, 1/\eta + ir]\cap [0,R]$ for integer $i\geq 0$. For $t\leq 1/\eta$,  we have $\Phi(t) \geq 1 > e^{\eta t -1}$. Thus, the inductive hypothesis holds
for $i=0$. Suppose, it holds for $i$, we prove it for $i+1$. 

First, consider an arbitrary $t^*\in [1/\eta, 1/\eta + (i+1)r]\cap [0,R]\setminus (Y_0 + r)$, where $Y_0+r$ is the set $Y_0$ shifted right by $r$.
Let $t=t^*-r$. Note that $t > 0$, since $r < \nicefrac{1}{\eta}$. Also, $t\notin Y_0$.
Then, by the inductive hypothesis, we have $\Phi(x)\geq e^{\eta x-1}$ for all $x\in [1/\eta,t]$.
Using Inequality~(\ref{ineq:Phi}), we obtain the following bound
\begingroup
\allowdisplaybreaks
\begin{align*}
\Phi(t^*) &= \Phi(t+r)
\\ 
&\geq \Phi(t) + \gamma \int_0^{1/\eta} \Phi(x) dx +
\gamma \int_{1/\eta}^{t} \Phi(x) dx 
\\
&\geq e^{\eta t-1} + \gamma \int_0^{1/\eta} 1\, dx +\gamma \int_{1/\eta}^{t} e^{\eta x-1} \,dx\\
&= e^{\eta t - 1} + \nicefrac{\gamma}{\eta}  + \nicefrac{\gamma}{\eta}\cdot (e^{\eta t-1} - 1)\\
&= e^{\eta t - 1}(1+ \nicefrac{\gamma}{\eta}).
\end{align*}
\endgroup
Since $\eta = \sqrt{\nicefrac{\gamma}{(e-1)r}}$, we have $\nicefrac{\gamma}{\eta} = (e-1)\eta r$. Now, using the
inequality $e^x\leq 1 + (e-1)x$ for $x\in [0,1]$, we get
\begin{align*}
\Phi(t^*)& \geq e^{\eta t - 1}(1+ \nicefrac{\gamma}{\eta}) =e^{\eta t - 1}(1+ (e-1)\eta \delta)\\
&\geq
e^{\eta t - 1}\cdot e^{\eta r} = 
e^{\eta (t+r)-1}=e^{\eta t^*-1}.
\end{align*}
To finish the proof, we need to show that 
$\Phi(t^{**})\geq e^{\eta t^{**}-1}$ for 
$t^{**}\in [1/\eta, 1/\eta + (i+1)r]\cap [0,R]\cap (Y_0 + r)$. Since $Y_0+r$ has measure zero, there exists an increasing sequence $t^*_k$ of numbers in $[0, 1/\eta + (i+1)r]\cap ([0,R]\setminus (Y_0 + r))$  that tends to $t^{**}$ as $k\to \infty$.
Using that $\Phi$ is a non-decreasing function and $e^{\eta t -1}$ is a continuous function, we have
$$\Phi(t^{**})\geq \lim_{k \to \infty} \Phi(t^*_k) \geq \lim_{k \to \infty} 
e^{\eta t^{*}_k-1} = 
e^{\eta t^{**}-1}.$$
\end{proof}
We now show that Lemma~\ref{lem:Phi} implies Lemma~\ref{lem:Phi-S}. 
Loosely speaking, in the proof, we shift all intervals from the set $S$ to the left to obtain a single interval
$[0,\mu(S)]$. We then apply Lemma~\ref{lem:Phi} to the transformed function.
\begin{proof}[Proof of Lemma~\ref{lem:Phi-S}]
Let $\pi_S$ and $\pi_S^{inv}$ be the maps defined in Section~\ref{sec:light-ball}.
Define $\Phi^{*}(t)$ as $\Phi^{*}(t) = \Phi(\pi_S^{inv}(t))$
and let $Y_0$ be a measure zero set as in Lemma~\ref{lem:Phi-S}.
We claim that $\Phi^{*}(t)$ satisfies (\ref{ineq:Phi}) for all $t\in [0,\pi(S)]\setminus Y_0$. Fix 
$t\in [0,\pi(S)]\setminus Y_0$. Write
$$\Phi^*(t+r) = \Phi(\pi^{inv}_S(t+r)) \geq
\Phi(\pi^{inv}_S(t)+r).
$$
Here, we used that (a)
$\pi^{inv}_S(t+r)\geq \pi^{inv}_S(t)+r$ and (b)
$\Phi$ is a monotone function. By
Lemma~\ref{lem:Phi-S}, $\pi^{inv}_S(t)\in S$, thus
\begin{align*}
\Phi^*(t+r) &\geq \Phi(\pi^{inv}_S(t)+r)\\ &\geq 
\Phi(\pi^{inv}_S(t)) + \gamma \int_0^{\pi^{inv}_S(t)} \Phi(x) dx.
\end{align*}
We now observe that $\Phi^*(t) = \Phi(\pi^{inv}_S(t))$ and 
\begin{align*}
\int_0^{\pi^{inv}_S(t)} \Phi(x) dx &\geq 
\int_0^{\pi^{inv}_S(t)} \Phi(x) \cdot \one(x\in S) dx\\
&=
\int_0^{\pi^{inv}_S(t)} \Phi(x) d \pi_S(x) \\
&= \int_0^{t} \Phi^*(x) dx.
\end{align*}
Here, we used that $d\pi_S(x) = \one(x\in S) dx$ and
$\pi_S(\pi^{inv}_S(t))=t$.
Thus, we showed that for all $t \in [0,\pi(S)]\setminus Y_0$, we have 
$$
\Phi^*(t+r) \geq \Phi^*(t) + \int_0^{t} \Phi^*(x) dx.
$$
We now use Lemma~\ref{lem:Phi} with function $\Phi^*$ and $R'=\mu(S)$. We obtain the following inequality:
$$\Phi^*(\mu(S))\geq e^{\eta \mu(S)-1},$$
which concludes the proof of Lemma~\ref{lem:Phi-S}.
\end{proof}

\section{Integrality Gap}\label{Gap_section}
%\ynote{Maybe we can simly present this integrality gap for $k=1$? The construction and analysis are simpler.}
In this section, we present an integrality gap example for the convex program $(P)$ in Figure~\ref{fig:CP}. 

% \ \hrule\\

%\ynote{The construction for $k=1$}
\paragraph{Construction.}
Let $n = 1+ \lceil\sqrt{1/\alpha}\rceil$. Consider a complete graph on $n$ vertices. Let $P$ be a path of length $n-1$. Denote its endpoints by $s$ and $t$ and the set of its edges by $E_P$. All edges in $P$ are positive edges of weight 1. Edge $(s,t)$ is a negative edge of weight 1. All other edges are positive edges of weight $\alpha$. 

\paragraph{The cost of the integral solution} Clearly, every integral solution $\cal P$ should violate some edge $(u,v)\in P \cup \{(s,t)\}$ (since all these edges cannot be satisfied simultaneously).
Thus, $\disagree_u(\calP, E^+, E^-) \geq 1$ and $\|\disagree(\calP, E^+, E^-)\|_p \geq 1$.

\paragraph{The cost of the CP solution.}
We define the CP solution as follows.  
Denote the distance between $u$ and $v$ along $P$ by $\mathrm{dist}_P(u,v)$. Let $x_{uv} = \mathrm{dist}_P(u,v)/(n-1)$. 
Note that $x_{st} = 1$.
The values of variables $y_u$ are determined by constraints (P1) of the convex program.

Now we upper bound the contribution of every edge $(u,v)$ (incident on $u$) to $y_u$ in formula (P1).
The contribution of $(u,v)\in E_P$ is $w_{uv} x_{uv} = 1 \cdot 1/(n-1)$;
the contribution of edge $(s,t)$ is
$w_{st}(1-x_{st}) = 0$ (whether or not it is incident on $u$), and the contribution of every other edge $(u,v)$
is $w_{uv}x_{uv} \leq \alpha$.
Since every vertex $u$ is incident on at most 2 edges from $E_P$, we have
$y_u \leq 2 /(n-1) + \alpha n \lesssim \sqrt{\alpha}$. Now, 
$$\|y\|_p \leq n^{1/p} \cdot \max_u |y_u| \lesssim n^{1/p} \alpha^{1/2} \lesssim {\alpha}^{1/2-1/(2p)}.$$

\paragraph{Integrality gap}
We conclude that the integrality gap is at least $\Omega((\nicefrac{1}{\alpha})^{1/2-1/(2p)})$.

\end{document}